\documentclass[fullpage, 11pt]{article}
\usepackage{amsfonts,amsbsy,amssymb,amsmath,graphicx}
\usepackage{epsfig}
\usepackage{dsfont}
\usepackage{natbib}
\usepackage{mathrsfs}
\usepackage{geometry}
\usepackage{multirow}

\usepackage{color}
\usepackage{natbib}
\usepackage{rotating}
\usepackage{epsfig}
\usepackage{ifpdf}
\usepackage[colorlinks=true,linkcolor=red,citecolor=blue]{hyperref}

\newenvironment{proof} {\noindent {\em \textbf{Proof}} } { \hfill \fbox{~} \\ }

\def\1{1\kern-.20em {\rm l}}

\newtheorem{theorem}{Theorem}[section]

\newtheorem{corollary}[theorem]{Corollary}

\newtheorem{lemma}[theorem]{Lemma}

\newtheorem{proposition}[theorem]{Proposition}
\newtheorem{remark}[theorem]{Remark}

\numberwithin{equation}{section}

\newcommand{\R}{\mathbb{R}}

\newcommand{\G}{\mathcal{G}}
\newcommand{\F}{\mathcal{F}}

\newcommand{\E}{\mathbb{E}}

\def\1{1\kern-.20em {\rm l}}

\numberwithin{equation}{section}

\begin{document}
\title{\bf Nonparametric  $M$-estimation for right censored regression model
 with stationary ergodic data}

\author{
Mohamed Chaouch$^{1,}\footnote{corresponding author}$ ~~
Na\^amane La\"ib$^{2, \,3}$ ~~and ~~
Elias Ould Sa\"id$^{4,\,5}$\\
$^1$Department of Statistics\\ United Arab Emirates University, UAE\\
$^2$L.S.T.A., Universit\'e Paris VI, Paris - France\\
$^3$EISTI, avenue du parc 95011, Cergy-Pontoise, France.\\
$^{4}$Univ. Lille Nord de France, F-59000 Lille - France \\ $^{5}$ULCO, LMPA,  CS: 80699 Calais - France
}

\maketitle

\begin{abstract}
\noindent The present paper deals with  a nonparametric  $M$-estimation for right censored regression model with stationary ergodic data.
 Defined as an implicit function, a kernel type estimator of a family of robust regression is considered when the covariate take its values in $\mathbb{R}^d$ ($d\geq 1$) and the data are sampled from {\it stationary ergodic process}.
The strong consistency (with rate) and the asymptotic distribution of the estimator are established under mild assumptions. Moreover, a usable confidence interval  is provided which  does not depend on any unknown quantity. Our results hold without any mixing condition and do not require the existence of marginal densities. A comparison study based on simulated data is also provided.
\vspace{8mm}

\vspace{5mm}

\end{abstract}

\section{Introduction}\label{intro}

\noindent  Consider a pair $(X,T)$ of random variables defined in $\mathbb{R}^d\times\mathbb{R}$, $d\geq 1$, where $T$ is a variable of interest and $X=(X^1, \dots, X^d)$ a vector of concomitant variables. In many situations, one can be interested in the regression function  $m(x)=\mathbb{E}(T |X=x)$  for $x\in\mathbb{R}^d$,
 which allows to describe the relationship between T and X.



 \noindent Let $\rho(\cdot)$ be a continuous differentiable function such that $\partial \rho(\cdot)/\partial \theta =: \psi(\cdot)$ is supposed to be nondecreasing real-valued function. Let us denote by $m_\psi(\cdot)$ the location parameter  (or equivalently the $\psi$-regression function) that defined as  a solution, with respect to $\theta$,  of the following minimization problem
\begin{eqnarray}\label{exp1}
\min_{\theta}\mathbb{E}\left[\rho(T-\theta) | X=x\right].
\end{eqnarray}
  We call such a $m_\psi(\cdot)$ the $M$-functional corresponding to $\psi(\cdot)$. Notice that whenever $\rho(\cdot)$ is a regular function,
  $m_\psi(\cdot)$ may be considered as a solution, with respect to, $\theta$ of the equation

  \begin{equation}\label{psi_regression}
\mathbb{E}\left[\psi(T-\theta)\left|\right.X=x\right]=0.
\end{equation}


\noindent  We point out that the family of the robust functions $\psi(\cdot)$
may depend on $x$. This may occur when dealing with nonlinear regression models with ARCH-errors, in such case the function $\psi(\cdot)$
depends in the
 function $\sigma(x)$, which measures the spread  of the conditional law of $T$ given $X=x$.

\vskip 2mm
\noindent  Let $(X_i, T_i)_{i=1,\dots,n}$ be $n$ copies of the pair $(X, T)$. A kernel type estimator, say $\widehat{m}_{\psi}(x)$, of $m_{\psi}(x)$ is a solution of the equation
\begin{eqnarray}\label{equ5}
\sum_{i=1}^n w_{n,i}(x)\psi(T_i - \theta)=0,
\end{eqnarray}
\noindent where $w_{n,i}(x):= K(h_n^{-1}(X_i-x))/\sum_{i=1}^nK(h_n^{-1}(X_i-x))$ are the Nadaraya-Watson (NW) weights. Here, $K(\cdot)$ is a real-valued kernel function and $h:= h_{n}$  is a sequence of positive real numbers  that  decreases to zero as $n$ tends to infinity.

%

\noindent If we consider $\rho(u) = u^2$ (and thus $\psi(u)=u$), then $m_\psi(x)$ coincides with the classical regression function $m(x)$ and $\widehat{m}_{n,\psi}(x)$ is its NW estimator.  When $\rho(u) = |u|$ (and therefore $\psi(u)=\mbox{sign}(u)$), then
$\widehat{m}_{n,\psi}(x)$ corresponds to the local least absolute distance estimator. For more possible choices of the function $\psi(\cdot)$ the reader is referred to \cite{SER80}, page 246.

In the statistical literature, several papers have been devoted to the study the properties of the nonparametric $M$-estimator  defined as a solution of equation (\ref{equ5}) when the variable of
interest $T$ is completely observed.
One can refer, among others,  to \cite{HU64} and \cite{HA84} for  independent and identically distributed (i.i.d.) case, \cite{CH86} and \cite{BF89} for mixing processes and \cite{LO00}  for stationary ergodic processes. \cite{ALO08} established the almost complete convergence rate of the kernel type estimator in the i.i.d. case  while
 \cite{CR92} studied its asymptotic properties under  $\alpha$-mixing assumption.  \cite{CR08}  studied  the same problem for {\it functional} covariate. They  established explicit exact asymptotic expression of the convergence rate in $\mathbb{L}^p$ norm. While \cite{GUE13} established the almost complete convergence with rate in the setting of functional and  stationary  ergodic data. However, in many practical situation such as  in medical follow-up or in engineering life-test
study, the variable of interest $T$ may not be completely observable. This case may be occur when dealing with censored data.
 For example, in the
clinical trials domain, it frequently happens that patients from
the same hospital have correlated survival times due to unmeasured
variables such as the quality of hospital equipment. For further real practical examples, the reader can be referred (among others) to \cite{WLW89}  for the regression method with applied data,  and \cite{LI00} for the case of missing covariates.


\vskip 2mm

\noindent The regression model in presence of {\it censored data} has been studied by several authors. For instance \cite{CO72}  considered the linear regression model and estimated the slope via the proportional hazard model. In the general linear case, many approaches have been used, see for instance
\cite{KSV81}  and \cite{KS98} . \cite{RG97} proposed a type $M$-estimators for the linear regression model with random design when the response observations are doubly censored. For nonlinear regression model, \cite{BER81} introduced a class of nonparametric estimators for the conditional survival function in the presence of right-censoring. He proved
some consistency results of these estimates, his work has  extended by \cite{DA87, DA89}. \cite{JI07} constructed an  $M$-estimators for the regression parameters in semi parametric linear models for censored data and established asymptotic normality of these estimators.

\vskip 2mm
\noindent  Notice that, the estimator of the conditional survival function could be used in order to get a consistent estimate of the regression function $m(\cdot)$ in the presence of {\it censored data}. However,  the computations may be difficult practically. To overcome this drawback, \cite{CGV95} introduced   a general  nonparametric partitioning estimate of $m(\cdot)$ and proved its strong consistency. \cite{KMP02}  gave a simpler proof for kernel, nearest neighbor, least squares and penalized least squares estimates. 


\noindent  In all papers mentioned above when dealing with {\it dependent data} the condition of  $\alpha$-mixing is assumed to be fulfilled.  A large  class of processes satisfy this condition.
However, there are still a great number of models where such assumption does not hold
(see \cite{LL10} for some examples). It is then necessary to consider a general larger dependency framework as is the ergodicity.

\vskip 3mm
For the sake of clarity, introduce some details defining the ergodic property of processes and its link with mixing one.
Let $\{ X_n, n\in \mathbb{Z} \}$ be a stationary sequence. Consider the backward field ${\cal B}_n=
\sigma(X_k; k\leq n)$ and the forward field ${\cal F}_m=\sigma(X_k; k\geq m)$. The sequence is strongly mixing if
$$ \sup_{A\in {\cal B}_0, B\in {\cal B}_n}| \mathbb{P} (A \cap B)-\mathbb{P} (A) \mathbb{P} (B)|=\varphi(n)\to 0 \quad \mbox{as}\quad n\to \infty.$$
The sequence is ergodic if

\begin{eqnarray}\label{ergodic}
 \lim_{n\to\infty} \frac {1}{n} \sum_{k=0}^{n-1} |\mathbb{P} (A \cap \tau^{-k}B) - \mathbb{P}(A) \mathbb{P}(B)| = 0,
 \end{eqnarray}
where $\tau$ is the time-evolution or shift transformation. The naming of strong mixing in the above definition  is a more stringent condition than that what ordinarily referred (when using the vocabulary of measure-preserving dynamical systems) to as strong mixing, namely that $\lim_{n\to\infty} \mathbb{P} (A \cap \tau^{-n}B) = \mathbb{P}(A)\mathbb{P}(B)$  for any two measurable sets $A, B$ (see \cite{R72}).
Hence, strong mixing implies ergodicity. However, the converse is not true: there exist ergodic sequences  which are not strong mixing.
The ergodicity condition is then a naturel condition and lower than any type of mixing for which usual nonparametric estimators
(density, regression, ...) are convergent. It seems to be a condition of obtaining large numbers of law, since
it is well known from the ergodic theorem that, for a stationary
ergodic process $Z$, we have
\begin{eqnarray}\label{ergodicity}
\lim_{n\to\infty}\frac{1}{n}\sum_{i=1}^n Z_i=\mathbb{E}(Z_1), \ \mbox{almost surely (a.s.)}
\end{eqnarray}

The ergodic property in our setting is formulated on the basis of the statement (\ref{ergodicity}) and the requirements are
considered in conditions (A2) below. We refer to the book of \cite{K85} for an account of details and results on the ergodic theory.

\vskip 2mm

 Recently, and in the case of complete data, \cite{LL11}  studied the asymptotic properties of
 the regression function using functional stationary ergodic data. In the case of right censored response, \cite{CK13}
considered the conditional quantile estimation based on functional stationary ergodic data.


\vskip 2mm
\noindent  In this paper, we are interested in the kernel smoothing
 estimation of the $\psi$-regression function for right-censored and  stationary ergodic data. For more motivation
concerning this kind of dependency  one may refer to \cite{LL11}.
We emphasis that, our results  hold without assuming any type of mixing conditions as well as the existence of
marginal and conditional densities. This is of particular interest for chaotic models where
the underlying process does not posses a marginal density
(see \cite{LU99} and \cite{LA05} for a discussion). To the best of our knowledge this question,
 under the general hypotheses cited above,  has not been studied in literature. The proof our  results used only technical of martingale differences combined with the ergodic property  given in (\ref{ergodicity}). We avoid  to suppose any additional condition on the structure of  the process under study  as in  the $\alpha$-mixing case.

\vskip 2mm
\noindent  The paper is organized as follows. In Section \ref{sec2}, a kernel type estimator of the $\psi$-regression is introduced for right censorship model under an ergodic assumption. Section \ref{sec3} presents the assumptions under which the almost sure consistency (with rate) and the asymptotic distribution of the estimator are established. A simulation study that shows the performance of our estimator is provided in Section \ref{sec4}. Proofs of main results and some auxiliary results (with their proofs) are postponed in Section \ref{sec5}.

\section{Robust regression estimation with censored data}\label{sec2}
Let us consider a triple $(X,C,T)$ of random variables defined in $\mathbb{R}^d\times\mathbb{R}\times\mathbb{R}$, where $T$ is the variable of interest (typically a lifetime variable), $C$ a censoring variable and $X=(X^1, \dots, X^d)$ a vector of covariates. We denote by $F(\cdot)$ (resp. $G(\cdot)$) the distribution function of $T$ (resp. $C$) which are supposed to be unknown and continuous. The continuity of $G$ allows to use convergence results for the \cite{KM58} estimator of $G$. From now on, we assume that
\begin{itemize}
\item[$(\bf A0)$] $(T,X)$ and $C$ are independent.
\end{itemize}
 This assumption plays an important role to derive the result given by (\ref{calsyn}) below. It  has been introduced by \cite{CGV95}
 and used (among others) by  \cite{KMP02} and \cite{GO08}. It is plausible whenever the censoring is independent of the characteristics of the patients under study.

\noindent  In the right {\it censorship model}, the pair $(T,C)$ is not directly observed and the corresponding available information is given
by $Y=\min(T,C)$ and $\delta=\1_{\{ T\leq C\}}$,
where $\1_A$ denotes the indicator function of the set $A$. Therefore, we assume that a sample $\{(X_i, \delta_i, Y_i), i=1, \dots, n\}$ is at our disposal.
Moreover, we suppose throughout the paper that $(X_i, T_i)_{i=1, \dots,n}$ is a strictly stationary ergodic sequence,  in the sense that satisfies the statement (\ref{ergodicity}),  and $(C_i)_i$ is a sequence of (i.i.d.) r.v's which is independent of $(X_i, T_i)_{i=1, \dots, n}$. One may observe that, by continuity of the identity application, the unobserved sample $(X_i, T_i, C_i)_{i=1, \dots, n}$ is ergodic. In addition let $\wp: \R^3 \rightarrow \R \times \{0,1\}\times\R$ be a measurable application defined by $(T,C,X) \mapsto (Y,\delta,X).$ Clearly, by the multidimensional ergodic theorem, the observed sample is stationary and ergodic one, as soon as the unobserved sample is.

\noindent  When dealing with censored data, many authors (see
{\it e.g.} \cite{CGV95}, \cite{KMP02}
  \cite{GO12} and \cite{LOS13} among others) use the
so-called {\it synthetic data} which allow to take into account the
censoring effect on the lifetime distribution.  To make clear this notion, let
\begin{equation}
Z=\frac{\delta Y}{\overline{G}(Y)}, \quad \mbox{where} \quad  \overline{G}(\cdot)=1-G(\cdot).
\label{donn�es-synth�tiques}
\end{equation}

\noindent
Making use of  double conditioning withe respect to $(X, T)$ combined with ({\bf A0}), one may write

\begin{eqnarray}\label{calsyn}
\mathbb{E}\left\{ Z | X\right\} = \mathbb{E}\left\{\frac{T}{\overline{G}(T)} \mathbb{E}\left(\1_{\{T\leq C\}} | X,T \right) | X \right\} = \mathbb{E}\{T|X\}.
\end{eqnarray}
Therefore, any estimator for the regression function $\mathbb{E}\left(Z | X=x\right)$, which can  be built on fully   observed data$(Y_i, C_i)$ , turns out to be an estimator for
the regression function $\mathbb{E}(T|X=x)$ based on the unobserved data.

\noindent  In order to define the $\psi$-regression function under the right censorship model, consider (as in (\ref{donn�es-synth�tiques})), the {\it synthetic} $\psi^*$-function defined as

\begin{eqnarray}\label{syntetique2}
\psi^*(T-\theta):=\frac{\delta\psi(T-\theta)}{\overline{G}(T)}.
\end{eqnarray}
\noindent The problem being find the parameter $\theta_{\psi}(x)$ which is a zero
 w.r.t. $\theta$ of
\begin{eqnarray}\label{regression-censure}
 \Psi(x,\theta):=\mathbb{E}\left[\frac{\delta\psi(T-\theta)}{\overline{G}(T)}\left|\right.X=x\right] =0.
\end{eqnarray}

\noindent Note here that the {\it synthetic} $\psi^*$-function, given by (\ref{syntetique2}), inherits the monotony
property from $\psi(T-\cdot)$.

\noindent Based on the observed sample $(X_i, \delta_i, Y_i)_{i=1, \dots, n}$, we define the following  ``pseudo-estimator" of $\Psi(x,\theta)$,
 which will be used as intermediate estimator
\begin{equation}
\widetilde{\Psi}_n(x,\theta):=\frac{\sum_{i=1}^n K(h^{-1}(x-X_i))\,  \delta_i \,(\overline{G}(Y_i))^{-1}\,\psi(Y_i-\theta)}{\sum_{i=1}^n K(h^{-1}(x-X_i))} =: \frac{\widetilde{\Psi}_{N}(x,\theta)}{\widehat{\Psi}_{D}(x)},
\label{psitilde}
\end{equation}
where
\begin{eqnarray}\label{psodoestimate}
\widetilde{\Psi}_{N}(x,\theta) := \frac{1}{n\mathbb{E}(\Delta_1(x))} \sum_{i=1}^n \Delta_i(x) \frac{\delta_i\, \psi(Y_i-\theta)}{\overline{G}(Y_i)}
\quad
\mbox{and}
\quad
\widehat{\Psi}_{D}(x) := \frac{1}{n\mathbb{E}(\Delta_1(x))} \sum_{i=1}^n \Delta_i(x),
\end{eqnarray}
with  $\Delta_i(x) = K(h^{-1}(x-X_i))$. In practice $G(\cdot)$ is unknown, therefore
to get  a feasible estimator, we  should replace
$\overline{G}(\cdot)$ by its \cite{KM58} estimator
$\overline{G}_n(\cdot)$  given by

$$\overline{G}_n (t) =\left\{
\begin{array}{ll}
\prod_{i=1}^n \left(1-\frac{1-\delta_{(i)}}{n-i+1}\right)^{\1_{\left\{Y_{(i)}\leq t\right\}}}
&\mbox{if } t< Y_{(n)}\\
0 & \mbox{otherwise},
\end{array}
\right.$$
where $Y_{(1)}<Y_{(2)}<...<Y_{(n)}$ are the order statistics
of $( Y_i)_{1\leq i\leq n}$ and $\delta_{(i)}$ is the concomitant
of $Y_{(i)}$.

\noindent Therefore an estimator of $\Psi(x,\theta)$ can be defined as
\begin{equation}
\widehat{\Psi}_n(x,\theta):=\frac{\widehat{\Psi}_{N}(x,\theta)}{\widehat{\Psi}_{D}(x)},
\label{psichap}
\end{equation}
where the denominator has the same expression as $\widetilde{\Psi}_{N}(x,\theta)$ by substituting $\overline{G}(\cdot)$
by   $\overline{G}_n(\cdot).$

\noindent Thus, Therefore, an   estimate of the $\psi$-regression function $\theta_{{\psi}}(x)$, say  $\widehat{\theta}_{{\psi},n}(x)$,
 may be defined as a zero w.r.t. $\theta$ of
$ \widehat{\Psi}_n(x, \theta)=0,$
which  satisfies
\begin{eqnarray}\label{defest}
\widehat{\Psi}_n(x, \widehat{\theta}_{{\psi},n}(x)) = 0.
\end{eqnarray}

\section{Main results}\label{sec3}

\noindent  In order to state our results, we introduce some notations. Let
${\mathcal F}_i$ be  the  $\sigma$-field generated by $(
(X_1,T_1), \ldots, (X_i, T_i))$ and  ${\mathcal G}_i$  the one
generated by $( (X_1,T_1), \ldots,$ $(X_i, T_i), X_{i+1})$. Set
$||.||_2$  for the  Euclidean norm  and $||.||$ the sup  norm in
$\mathbb{R}^d$. For any $x$ in  $\mathbb{R}^d$ and for  $r>0$,
denoted by $S_{r,x}:=\{u: ||u-x||\leq r\}$ the sphere of radius
$r$ centered at  $x$. For any   Borel set $A\subset \mathbb{R}^d$,
set $\mathbb{P}_{X_{i}}^{ {\cal F}_{i-1}}(A)=\mathbb{P}\left(X_{i}\in A \vert {\cal
F}_{i-1}\right)$.  For some $\tau>0$,  let ${\cal
C}_{\psi,\tau}:=[\theta_{\psi}(x)-\tau, \; \theta_{\psi}(x)+\tau]$  be   a subset
of $\mathbb{R}$ and  ${\cal V}_x$ be a neighborhood  of
$x$. Denote by $\mathcal{O}_{a.s.}(v)$ a real random function $\ell$ such that $\ell(v)/v$
is almost surely bounded as $v$ goes to zero. Furthermore, for any distribution function $L(\cdot)$, let $\zeta_L=
\sup\{t, \;\mbox{such that}\; L(t) <1 \}$ be the support's right endpoint.
Suppose that $\theta_{\psi}(x) \in {\cal C}_{\psi,\tau}\cap (-\infty, \zeta]$,
where $\zeta <\zeta_G\wedge\zeta_F.$ \\

\noindent Our results are stated under some  assumptions that we  gathered
hereafter  for easy reference.
\begin{itemize}

\item [({\bf A1})]

\noindent There  exist a constant $c_0$ and a nonnegative bounded  random
function $h_{i}(x, \omega)=:h_{i}(x)$ (resp. $h(x)$), $\omega\in\Omega$, defined
on $\mathbb{R}^d$, such that, for any $i\geq 1$,
\begin{eqnarray}\label{Nor1}
 \mathbb{P}_{X_{i}}^{ {\cal F}_{i-1}}\left( S_{r,x}\right)=c_0 h_{i}(x)r^d \  \ \
 \mathrm{a.s.} \ \ \ \mathrm{as} \ \ \ r\to 0,
\end{eqnarray}
\begin{eqnarray}\label{Nor2}
\mbox{and}\quad \mathbb{P}_{X_i}(S_{r,x}) = c_0 h(x) r^d \ \ \ \mathrm{as} \ \ \ r\to 0.
\end{eqnarray}
\item [({\bf A2})] For any $x\in \mathbb{R}^d$, and $j=1,2$,
$\lim_{n\to\infty} n^{-1}\sum_{i=1}^{n}h_{i}^j(x)=h^j(x)$ a.s.
\item  [({\bf A3})] $\{b_n\}$ is a non-increasing sequence of positive
constants such that 

(i) $b_n\rightarrow 0$ and $\log n/nb_n^d \rightarrow 0$ as $n \rightarrow \infty$.

(ii) $b_n^d \ln \ln n\rightarrow 0 \quad \mbox{and} \quad nb_n^{1+2\alpha_1 d} \rightarrow 0 \quad \mbox{as} \quad 
n\rightarrow \infty \quad \mbox{for some} \quad \alpha_1>0$
\item [({\bf A4})] (i) \;$K$ is a spherically symmetric density function with
a spherical bounded support. That is, there exists   $k:
\mathbb{R}^+ \to \mathbb{R}$ satisfying $k(0)>0,k(v)=0$ for $v>1$,
and   $K(x)=k(||x||_2)$. Furthermore,  for $j\geq 1$, $k^j(\cdot)$  is  of class $\mathcal{C}^1$.\\
(ii)\; For any $m\geq 1$, $\int_0^1 k^m(u) u^{d-1} du \leq 1/d.$
\item [({\bf A5})] The function $\Psi$ given by
(\ref{regression-censure}) is such that :

    (i) $\Psi(x, \cdot)$ is of class ${\cal  C}^1$ on
    ${\cal C}_{\psi,\tau}$.\\
    (ii)   For each fixed $\theta \in {\cal C}_{\psi,\tau}$, $\Psi(\cdot,
    \theta)$ is continuous at the point $x$.\\
    (iii)
    $\forall (\theta_1, \theta_2) \in {\cal C}_{\psi,\tau}^2$ and $\forall (x_1, x_2)
     \in {\cal V}_x^2$, the derivative with respect to $\theta$ of order $p$ ($p\in \{0,1\}$)
    $\Psi^{(p)}(\cdot, \cdot)$ of  $\Psi(\cdot, \cdot)$ satisfying
    $$|\Psi^{(p)}(x_1, \theta_1)- \Psi^{(p)}(x_2, \theta_2)| \leq
    \gamma_1\|x_1-x_2\|_2^{\alpha_1}+\gamma_2|\theta_1-\theta_2|^{\alpha_2},$$
    $\displaystyle \mbox{for some} \quad \alpha_1>0, \alpha_2>0 \quad \mbox{and some
    constants}\quad \gamma_1>0, \gamma_2>0.$
\item [({\bf A6})] The function $\psi$ is such that :

(i) For any fixed $\theta\in {\cal C}_{\psi,\tau}$  and any $
j\geq 1$
$$\mathbb{E}[({\psi}^{(p)}(T_i- \theta))^{j} | {\cal G}_{i-1}]=\mathbb{E}[({\psi}^{(p)}(T_i-\theta))^{j} |
X_i] <  \gamma_3 j! <\infty \quad a.s. \quad \mbox{with} \quad p\in\{0, 1\} \quad\mbox{and}\quad \gamma_3>0.$$
(ii) $\psi(\cdot)$ is strictly monotone,
bounded, continuously differentiable  such that:  $\forall \theta\in \mathbb{R}$,
\\ $|\psi'(\cdot- \theta)| > \gamma_4 >0$ with $\gamma_4$ is a positive constant.

\item[({\bf A7})] (i) For any $(x_1, x_2)\in \mathcal{V}_x^2$, the function
\begin{eqnarray}\label{functM}
\mathbb{M}(x,\theta) := \mathbb{E}\left(\frac{\psi^2(T-\theta)}{\overline{G}(T)} | X=x\right), \quad x \in \mathbb{R}^d
\end{eqnarray}
satisfies
$$
|\mathbb{M}(x_1,\theta) - \mathbb{M}(x_2,\theta)| \leq \gamma_5 \|x_1 - x_2 \|^{\alpha_3}
$$
 $\gamma_5$ and $\alpha_3$ are some nonnegative constants.

(ii) For $m =1,2$, $\mathbb{E}\left(|\delta_1 \overline{G}^{-1}(T) \psi(T-\theta)|^m \right) <\infty$
and for any $x\in \mathbb{R}^d$, the conditional variance of $\delta_1 \overline{G}^{-1}(T) \psi(T-\theta)$ given $X=x$ exists,
 that is

\begin{eqnarray}\label{varianCond}
\mathbb{V}(\theta|x) := \mathbb{E}\left[ \left( \delta_1 \overline{G}^{-1}(T) \psi(T-\theta) - \Psi(x,\theta)\right)^2 | X=x\right].
\end{eqnarray}
 \item[({\bf A8})](i) The conditional variance of $\delta_i \overline{G}^{-1}(T_i) \psi(T_i-\theta)$ given the $\sigma$-field $\mathcal{G}_{i-1}$ depends only on $X_i$, i.e., for any $i\geq 1$, $\mathbb{E}\left[ \left(\delta_i \overline{G}^{-1}(T_i) \psi(T_i-\theta) - \Psi(X_i, \theta)\right)^2 | \mathcal{G}_{i-1}\right] = \mathbb{V}(\theta|X_i)$ almost surely.\\
 (ii) For some $\varsigma> 0$, $\mathbb{E}\left(|\delta_1 \overline{G}^{-1}(T_1) \psi(T_1-\theta)|^{2+\varsigma} \right) <\infty$ and the function
 \begin{eqnarray}\label{VV2}
 \widetilde{\mathbb{V}}_{2+\varsigma}(\theta|u) := \mathbb{E}\left[ \left(\delta_i \overline{G}^{-1}(T_i) \psi(T_i- \theta) - \Psi(X_i, \theta)\right)^{2+\varsigma} | X_i = u\right],\quad u\in \mathbb{R}^d\end{eqnarray} is continuous in $\mathcal{V}_x.$
 \end{itemize}

\begin{remark}\label{remark} Condition {\bf (A1)} means that the conditional
probability of the $d$-dimensional sphere, given the
$\sigma$-field  ${\cal F}_{i-1}$,  is asymptotically governed by a
local dimension when the radius $r$ tends to zero. This assumption
may be interpreted in terms of the fractal dimension, which
ensures that our results are established  without assuming the existence of
marginal and conditional densities (see  \cite{LU99} and \cite{LA05} for more details). It holds
true whenever the conditional distribution $\mathbb{P}_{X_{i}}^{ {\cal
F}_{i-1}}$  has a continuous conditional density $f_{X_i}^{{\cal
F}_{i-1} }(x)= h_{i}(x)$ at any point of the set $\{ x:
f_{X_{i}}^{ {\cal F}_{i-1} }(x)>0 \}$ and the constant  $c_0$
takes the value $\pi^{d/2}\Gamma((d+2)/2)$ with $\Gamma$  stands
for  the Gamma function.\\
Condition {\bf (A2)} is the equivalent to   C\'esaro's mean. The first assertion of ({\bf A3}) is used to establish the pointwise consistency rate of the estimator (see Theorem \ref{thm2}). While the second part of Condition {\bf (A3)} is used to vanish the bias term when dealing with the proof the asymptotic normality of the estimator (see proof of Theorem \ref{norm1}). Assumption {\bf (A4)}
concerns the kernel $k(\cdot)$. Condition {\bf (A5)} deals with some  regularities of the function $\Psi(\cdot)$. Condition {\bf (A6) (i)} is a standard assumption over $j^{th}$ moments of conditional expectation of the function $\psi(\cdot)$ and
{\bf (ii)} permits to ensure the existence and the uniqueness of the solution of (\ref{regression-censure}) (see \cite{CH86}).
Assumption ({\bf A7})(i) is a classical condition which allows to obtain the asymptotic normality. On the other hand ({\bf A7})(ii) is a necessary condition to define the conditional variance given in (\ref{varianCond}). ({\bf A8})(i) is a Markov condition and ({\bf A8})(ii) guarantee the existence of the quantity defined by (\ref{VV2}). 
\end{remark}

\begin{remark}
In order to check the assumptions ({\bf A5})-({\bf A8}), one can consider as an example the function $\psi(u)= \frac{u}{\sqrt{1+u^2}}$, where $u\in [-1,1]$. This function is bounded, of class ${\cal C}^1$ and non-decreasing. Moreover, assume that the censored variables $(C_i)_{i=1, \dots, n}$ are i.i.d. with exponential distribution ${\cal E}(\lambda)$, where $\lambda >0$, and assume that the random variables $T_i's$ are linked to the covariates $X_i's$ as follows $T_i = m(X_i) + \epsilon_i$, where $\epsilon_i$ are independent with an absolutely density function $f_\epsilon.$
\end{remark}



\subsection{Pointwise consistency with rate}
The following result gives a uniform  approximation (with rate) of the estimator $\widehat{\Psi}_n(x, \theta)$ by the pseudo-estimator of $\Psi(x, \theta)$, which plays an instrumental role
to proof the almost uniform sure consistency of $\widehat{\Psi}_n(x, \theta)$. All our results are state for sufficiently large $n$.

\begin{proposition}\label{lem2prop1}
Assume that assumptions ({\bf A1})-({\bf A6}) hold true, then we have
$$
\sup_{\theta \in {\cal C}_{\psi,\tau}}\Big| \widehat{\Psi}_n(x, \theta)  - \widetilde{\Psi}_n(x,\theta)\Big| =
\mathcal{O}_{a.s.}\left( \sqrt{\frac{\log\log n}{n}}\right).
$$
\end{proposition}

\noindent Theorem below provides the almost sure consistency with rate of the estimator
 ${\Psi}_n(\cdot, \cdot)$  uniformly  w.r.t. the second component.

\begin{theorem}\label{prop1}
Assume that ({\bf A1})-({\bf A6}) hold true, we have, for some $\alpha_1>0$, that
$$
\sup_{\theta \in {\cal C}_{\psi,\tau}}\Big|  \widehat{\Psi}_n(x, \theta)-\Psi(x, \theta) \Big|
  = \mathcal{O}_{a.s.}(b_n^{d \alpha_1}) + \mathcal{O}_{a.s.}\left( \sqrt{\frac{\log n}{n b_n^d}}\right).
$$
\end{theorem}

\noindent The following result state and prove the pointwise consistency with rate of the $M$-estimator $\widehat{\theta}_{\psi,n}(x)$.

\begin{theorem}\label{thm2}
Suppose that ({\bf A1})-({\bf A6}) are satisfied, then $\widehat{\theta}_{\psi,n}(x)$
 exists and is unique a.s., and
 \begin{eqnarray*}
\widehat{\theta}_{\psi,n}(x)-\theta_{\psi}(x)=\mathcal{O}_{a.s.}(b_n^{d \alpha_1})
+ \mathcal{O}_{a.s.}\left( \sqrt{\frac{\log n}{n b_n^d}}\right).
\end{eqnarray*}
\end{theorem}

\begin{remark}
If we choose $b_n=O\left(\left(\frac{\log n}{n}\right)^\tau\right)$ for some $\tau>0$, then the convergence rate given in Theorem below being $O\left(\frac{\log n}{n}  \right)^{\frac{1-\tau d}{2}}$ with $\frac{1}{d(2\alpha_1+1)}<\tau<\frac{1}{d}$.
\end{remark}
\subsection{Asymptotic distribution}

\noindent Theorem below gives the asymptotic distribution of the kernel type $M$-estimator  $\widehat{\Psi}_n(x, \theta)$.
 \begin{theorem}\label{norm1}
 Assume that Assumptions ({\bf A1})-({\bf A8}) hold true, then we have

 
 $$
 \sqrt{nb_n^d} \left( \widehat{\Psi}_n(x, \theta) - \Psi(x, \theta)\right)  \stackrel{\mathcal{D}}{\longrightarrow} \mathcal{N}(0, \sigma^2(x,\theta)),\qquad\mbox{as}\qquad  n\longrightarrow+\infty
 $$
 where $ \stackrel{\mathcal{D}}{\longrightarrow}$ denotes the convergence in distribution, $\mathcal{N}(\cdot, \cdot)$ the normal distribution and
 \begin{eqnarray}\label{variance1}
 \sigma^2(x,\theta) := \left[\mathbb{M}(x, \theta) - \Psi(x, \theta)\right]\times \frac{\int_0^1 k^2(u) u^{d-1} du}{h(x) d \left(\int_0^1 k(u)u^{d-1} du\right)^2}.
 \end{eqnarray}
 \end{theorem}

\noindent The following Theorem deals with the asymptotic distribution of the $\widehat{\theta}_{\psi,n}(x)$.
 \begin{theorem}\label{norm2}
 Under Assumptions ({\bf A1})-({\bf A8}), one gets
 $$
 \sqrt{nb_n^d} \left( \widehat{\theta}_{\psi,n}(x) - \theta_\psi(x)\right)  \stackrel{\mathcal{D}}{\longrightarrow} \mathcal{N}(0, \Sigma^2(x,\theta_\psi(x))),  \quad\mbox{as}\quad  n\longrightarrow+\infty
 $$
 where
 \begin{eqnarray}\label{variance2}
 \displaystyle\Sigma^2(x,\theta_\psi(x)) := \frac{\mathbb{M}(x,\theta_\psi(x))}{\left(\Gamma_1(x, \theta_\psi(x))\right)^2}\times \frac{\int_0^1 k^2(u) u^{d-1} du}{h(x) d \left(\int_0^1 k(u)u^{d-1} du\right)^2}
 \end{eqnarray}
 and
 \begin{eqnarray}\label{variance3}
 \displaystyle\Gamma_1(x,\theta) = \mathbb{E}\left[\psi^\prime(T-\theta) | X=x\right].
 \end{eqnarray}

 \end{theorem}

\subsection{Confidence intervals}

\noindent Notice that Theorem \ref{norm2} is useless in practice since many quantities in the variance are unknown. By Assumption ({\bf A1}), the term $h(x) b_n^d$ can be interpreted as the value of the probability that $X_i$ belongs to the sphere of radius $b_n$ and centered at $x$. In practice this probability might be estimated by $\widehat{\mathbb{P}}_{X_i} (S_{b_n, x}) := \frac{1}{n} \sum_{i=1}^n \1_{\{X_i \in S_{b_n,x}\}}.$ Moreover, $\mathbb{M}(x, \theta_\psi(x))$ and $\Gamma_1(x, \theta_\psi(x))$ could be replaced, in practice, by their nonparametric estimators, defined by:
$$
\widehat{\mathbb{M}}_n(x, \widehat{\theta}_{\psi,n}(x)) = \frac{\sum_{i=1}^n \psi^2(Y_i-\widehat{\theta}_{\psi,n}(x)) (\overline{G}_n(Y_i))^{-1} K(h^{-1}(x-X_i))}{\sum_{i=1}^n K(h^{-1}(x-X_i))}
$$
and
$$\widehat{\Gamma}_{1,n}(x, \widehat{\theta}_{\psi,n}(x)) = \frac{\sum_{i=1}^n \psi^\prime(Y_i-\widehat{\theta}_{\psi,n}(x)) K(h^{-1}(x-X_i) }{\sum_{i=1}^n K(h^{-1}(x-X_i))}.
$$
\noindent Hereafter, a Corollary that provides another form of the asymptotic distribution of $\widehat{\theta}_{\psi,n}(x)$ which can be used in practice to build confidence intervals.

\begin{corollary}\label{corr}
Under the assumptions of Theorem \ref{norm2}, one gets, as $n\rightarrow\infty$,
$$
\widehat{\Gamma}_{1,n}(x,\widehat{\theta}_{\psi,n}(x)) \left\{\frac{n\; \widehat{\mathbb{P}}_{X_i} (S_{b_n, x})\; d}{\widehat{\mathbb{M}}_n(x, \widehat{\theta}_{\psi,n}(x)) \int_0^1 k^2(u) u^{d-1} du} \right\}^{1/2}   \int_0^1 k(u) u^{d-1}du \times \left(\widehat{\theta}_{\psi,n}(x) - \theta_\psi(x) \right)   \stackrel{\mathcal{D}}{\longrightarrow} \mathcal{N}(0, 1).
$$
\end{corollary}

One can use the Corollary \ref{corr} to establish the $100(1-\alpha)\%$ confidence intervals for the $\psi$-regression $\theta_\psi(x)$.  It is defined, for any fixed $x\in\mathbb{R}^d$, as follows:
$$
\widehat{\theta}_{\psi,n}(x) \;\pm\; q_{\alpha/2}\; \widehat{\Gamma}_{1,n}(x,\widehat{\theta}_{\psi,n}(x)) \left\{\frac{n\; \widehat{\mathbb{P}}_{X_i} (S_{b_n, x})\; d}{\widehat{\mathbb{M}}_n(x, \widehat{\theta}_{\psi,n}(x)) \int_0^1 k^2(u) u^{d-1} du} \right\}^{1/2}   \int_0^1 k(u) u^{d-1}du,
$$
where $q_{\alpha/2}$ is the upper $\alpha/2$ quantile of the Gaussian distribution. All quantities appearing in  the confidence bands are known which make the
confidence interval  useful in practice. Notice  that the confidence interval is affected by the censorship model since it depends on $\overline{G}_n(\cdot)$.

\section{Simulation study}\label{sec4}
\subsection{Simulation 1: on the accuracy of the $M$-estimator}
In this section, we carry out a simulation to compare the finite sample performance of the $M$-estimator $\widehat{\theta}_{\psi,n}(x)$ and the NW estimator, say $\widehat{\theta}_{n}^\star(x)$  (see \cite{CGV95}) of the regression function when the response and the covariate are one-dimensional scalar random variables. Therefore, we consider the following process generated as follows: for $i=1, \dots, n$, $T_i = m(X_i)+ \sigma\epsilon_i$ and $X_i = 0.4 X_{i-1} + \eta_i$, where $\eta_i \sim \mbox{Bernoulli}(0.5)$, $ \epsilon_i \sim \mathcal{N}(0,1)$ and $\sigma = 0.01$. Here, three different regression models are considered:
\begin{itemize}
\item Model 1: $m(X_i) = X_i +2\exp{\{-16 X_i^2\}}$,
\item Model 2: $m(X_i) = X_i $,
\item Model 3: $m(X_i) = X_i ^2+1$.
\end{itemize}
Observe that, since $\eta_i$'s are Bernoulli distributed then the processes, described above, do not satisfy the $\alpha$-mixing condition whereas they are ergodic (see for instance \cite{LO00} and the references therein). The first regression model (Model 1) has been used by \cite{WL12} to study the accuracy of the M-estimator of the regression function when data are truncated and satisfy the $\alpha$-mixing assumption. Model 2 and Model 3 have been used by \cite{KLO10} and \cite{LOS13} to study the conditional mode and the robust regression estimators respectively under an $\alpha$-mixing and iid cases respectively. Here a more general dependence framework (ergodicity) is considered.

\noindent We also, simulate $n$ i.i.d. random variables $(C_i)_{i=1, \dots, n}$ with law $\mathcal{E}(5)$ (that is exponentially distributed with probability
density function $5e^{-5x}\1_{\{x\geq 0\}}$). Based on the observed data $(X_i, Y_i, \delta_i)_{i}$, we calculate our estimators by choosing $K(\cdot)$ the standard Gaussian kernel, and choose the function $\psi(u) = \frac{u}{\sqrt{1+u^2}}$. Observe that the function $\psi(\cdot)$ is bounded, of class $\mathcal{C}^1$ and non-decreasing function; therefore it satisfies the assumption (A6)(ii). In nonparametric statistics, it is well known that optimality, in the sense of MSE, is not seriously affected by the choice of the kernel but can be swayed by that of the bandwidth $h_n$. In this simulation study, two sample sizes are considered $n=300$ and 800 and for each sample size different Censoring Rates (CR) are taken $CR = 20\%, 40\%$ and $60\%$.

\noindent{\bf An empirical criteria for bandwidth selection}

\noindent  Furthermore, for a fixed sample size $n$ and a given censoring rate $CR$, a simple method was applied to choose the bandwidth. For a predetermined sequence of $h_n$'s from a wide range, say from 0.05 to 1 with an increment 0.02, we choose the bandwidth that minimises the {\it General Mean Square Error} (GMSE) of the estimators $\widehat{\theta}_{\psi,n}$ and  $\widehat{\theta}_{n}^\star$ respectively based on $B=500$ replications. We compute the GMSE of $\widehat{\theta}_{\psi,n}$ and  $\widehat{\theta}_{n}^\star$ as follows:
$
GMSE(\widehat{\theta}_{\psi,n}) = \frac{1}{B}\sum_{b=1}^{500}\frac{1}{n} \sum_{i=1}^n \left(\widehat{\theta}_{\psi,n}(X_i, b) - m(X_i, b)\right)^2,
$
and
$
GMSE(\widehat{\theta}_{n}^\star) = \frac{1}{B}\sum_{b=1}^{500}\frac{1}{n} \sum_{i=1}^n \left(\widehat{\theta}_{n}^\star(X_i, b) - m(X_i, b)\right)^2,
$
where $\widehat{\theta}_{\psi,n}(X_i, b)$ and $m(X_i, b)$ stand for the values of $\widehat{\theta}_{\psi,n}(x)$ and $m(x)$, respectively, at $X_i=x$ and the $b$-th replication.
\begin{table}[!t]
\caption{The global mean squared errors $\mbox{GMSE}(\widehat{\theta}_{\psi,n})$ and $\mbox{GMSE}(\widehat{\theta}_{n}^\star)$.}
\label{table1}
\centering

\begin{tabular}{lccccccccc}
\hline
$\mbox{CR}(\%)$ & n & Model 1& & & Model 2 &  & &Model 3 &  \\
\cline{3-10}
 & & $\widehat{\theta}_{\psi,n}$ &  $\widehat{\theta}_{n}^\star$ & & $\widehat{\theta}_{\psi,n}$ &  $\widehat{\theta}_{n}^\star$ & & $\widehat{\theta}_{\psi,n}$ &  $\widehat{\theta}_{n}^\star$\\
\hline
& 300 & 0.0967 & 0.2557 & &0.0263   & 0.0393 & & 0.0188 & 0.0505 \\
20 & & & & & & & & & \\
& 800 & 0.0945 & 0.3830 & & 0.0194& 0.0379& &0.019 & 0.048 \\
& & & & & & & & &\\
& 300 & 0.1612 & 0.2717& & 0.0633 &0.0648 & & 0.0284& 0.3719 \\
40 & & & & & & & & & \\
& 800 & 0.1753 & 0.3931& & 0.0356 & 0.0408 & & 0.0286 &  0.3198\\
& & & & & & & & & \\
& 300 & 0.2701 & 0.2995 & & 0.0661 & 0.0873 & &0.0356  & 0.428 \\
60 & & & & & & & & &\\
& 800 & 0.2834 & 0.42 & & 0.0613& 0.0828& & 0.0290 & 0.3233 \\
\hline
\end{tabular}
\end{table}

\noindent Table \ref{table1} depicts the values of the GMSE obtained for the $M$-estimator,
$\widehat{\theta}_{\psi,n}$, as well as the Nadaraya-Watson estimator $\widehat{\theta}_n^\star$
when different sample sizes and various censoring rates are considered for the three regression
models. One can observe, from Table \ref{table1}, that the two estimators perform better
when the sample size $n$ increases. It is also clear that the estimators' accuracy is
affected by the censoring rate, the quality of the estimators fell sharply when the censoring rate rose.

\begin{figure}[h!]
\begin{center}
\begin{tabular}{ll}
\includegraphics[height=8cm,width=8.5cm]{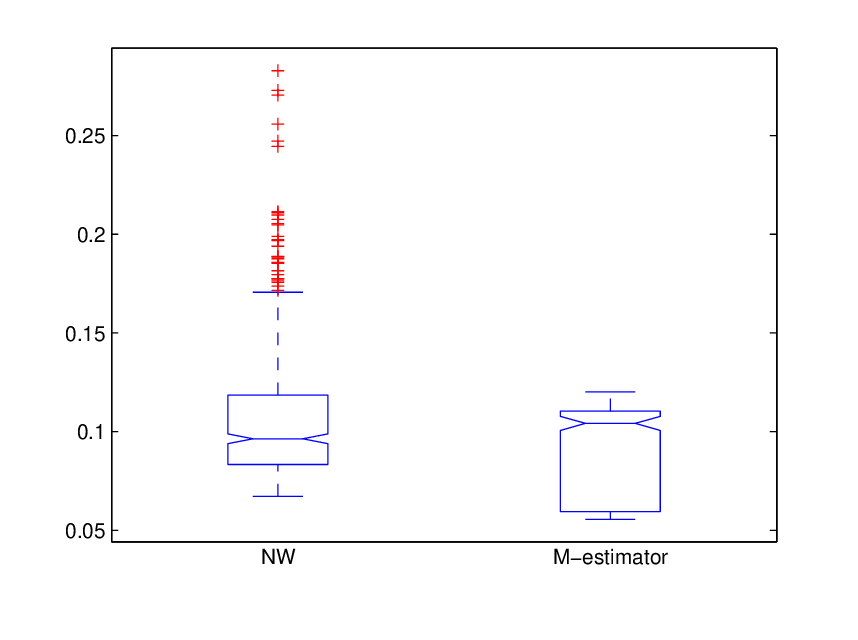} & \includegraphics[height=8cm,width=8.5cm]{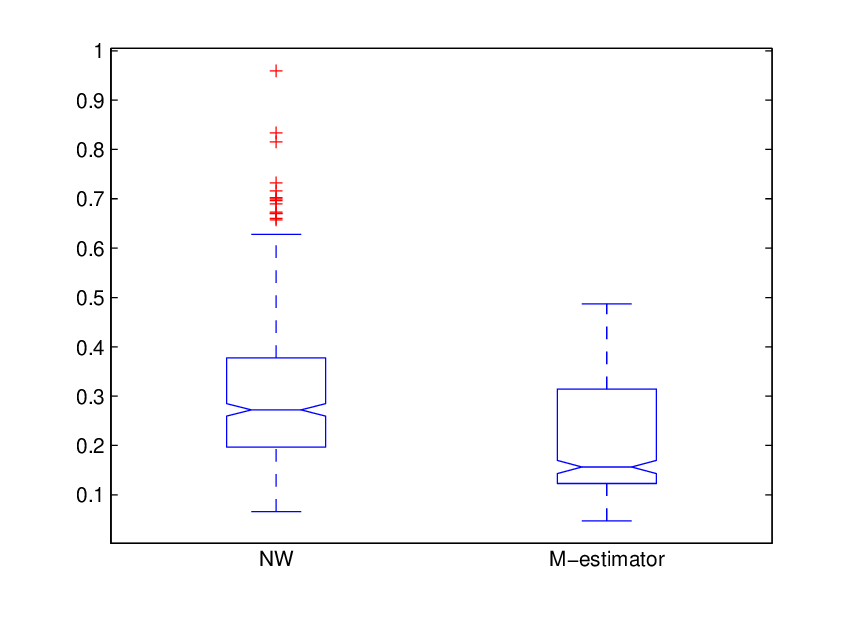} \\
\hspace{4cm}{\footnotesize{(a) $CR= 20\%$}} & \hspace{4cm}{\footnotesize{(b) $CR= 60\%$}}\\
\end{tabular}
\end{center}
\caption{Distribution of the $(MSE_b)_{b=1,\dots, B}$ obtained from $B=500$ replications using model 1 when $n=300$. }
\label{MSE20}
\end{figure}
\begin{figure}[h!]
\begin{center}
\begin{tabular}{ll}
\includegraphics[height=8cm,width=8.5cm]{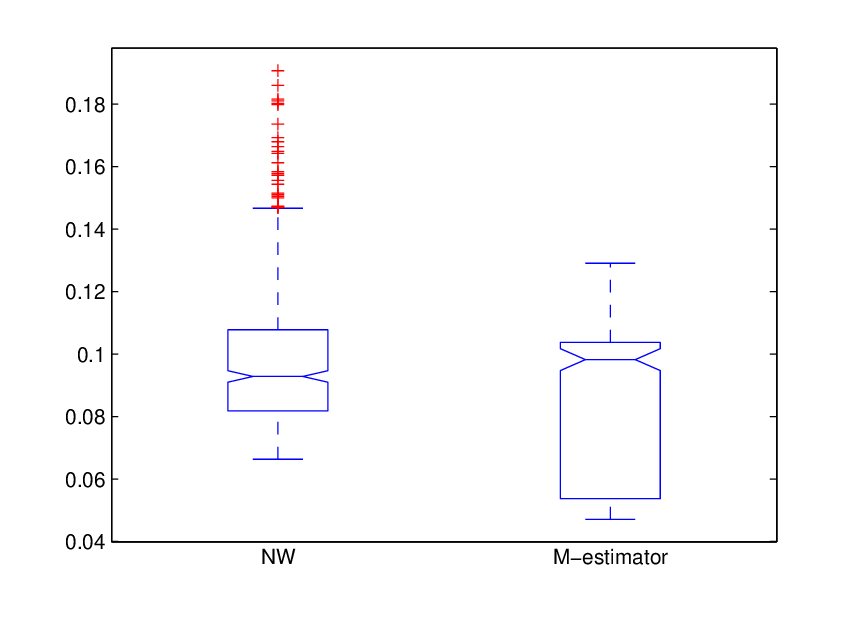} & \includegraphics[height=8cm,width=8.5cm]{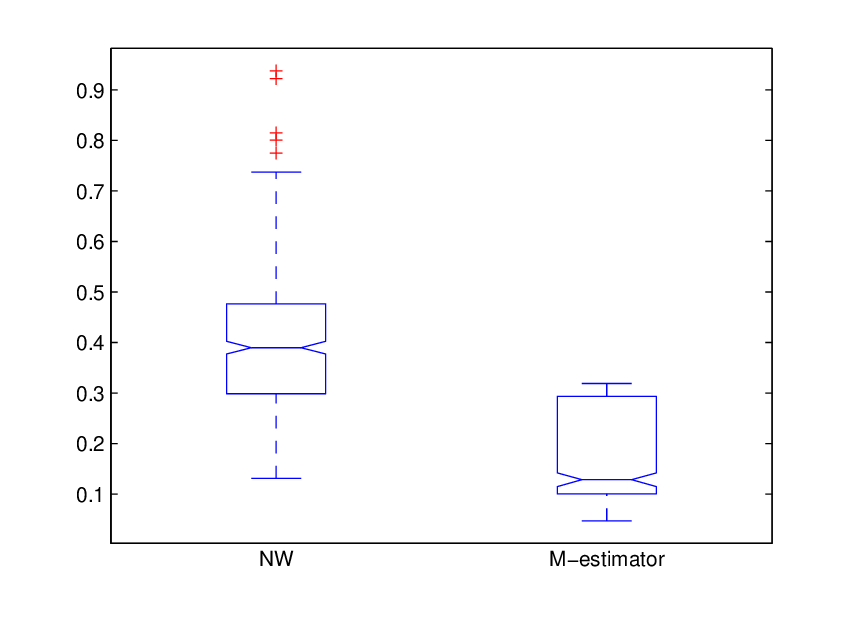} \\
\hspace{4cm}{\footnotesize{(a) $CR= 20\%$}} & \hspace{4cm}{\footnotesize{(b) $CR= 60\%$}}\\
\end{tabular}
\end{center}
\caption{Distribution of the $(MSE_b)_{b=1,\dots, B}$ obtained from $B=500$ replications using model 1 when $n=800$. } 
\label{MSE40}
\end{figure}

Moreover, in order to provide a deeper analysis of the errors, we consider particularly the model 1 and we plot in Figure \ref{MSE20} and \ref{MSE40} the distribution of the Mean Squared Errors of the two estimators obtained from the $B=500$ replications. As an overall trend, one can see that the $M$-estimator performs better than the NW estimator. This performance decreases when the censoring rate increases.

\begin{figure}[h!]
\begin{center}
\begin{tabular}{lll}
\includegraphics[height=7cm,width=5.5cm]{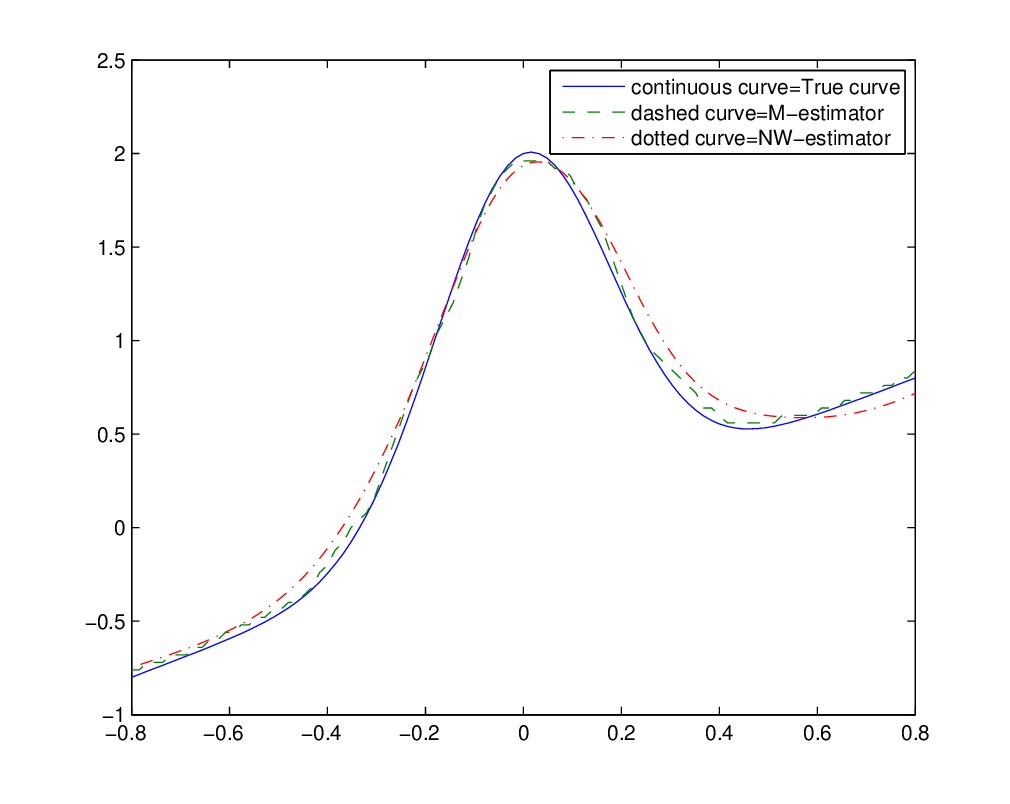} & \includegraphics[height=7cm,width=5.5cm]{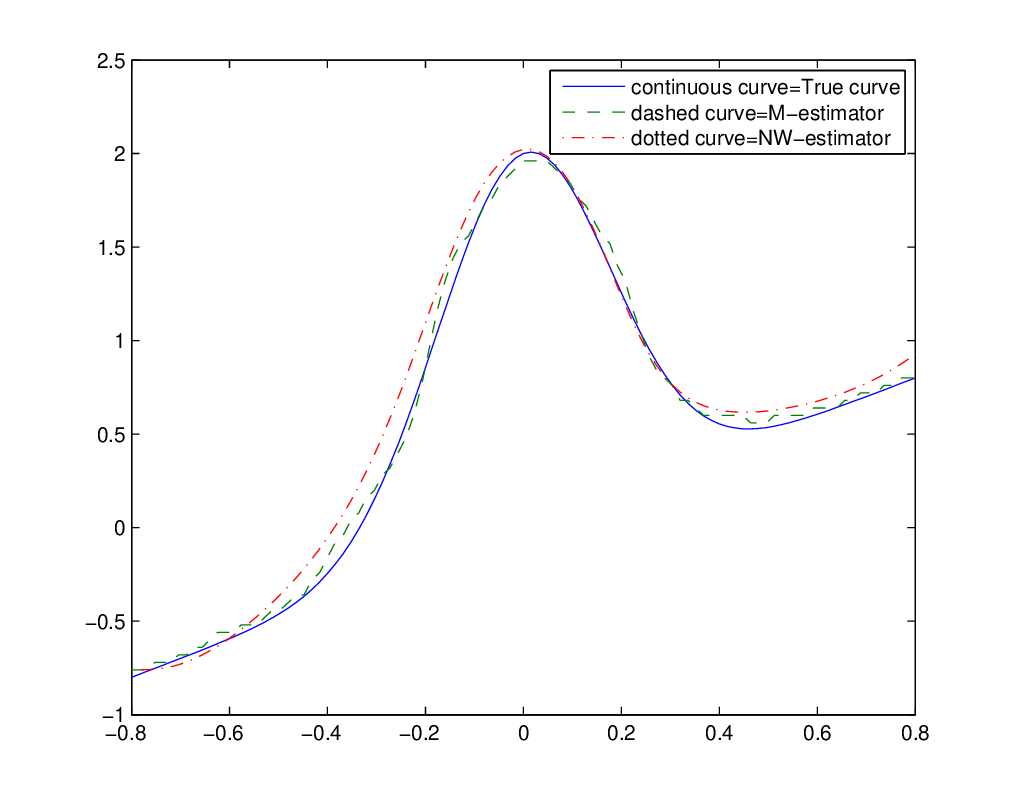}&
\includegraphics[height=7cm,width=5.5cm]{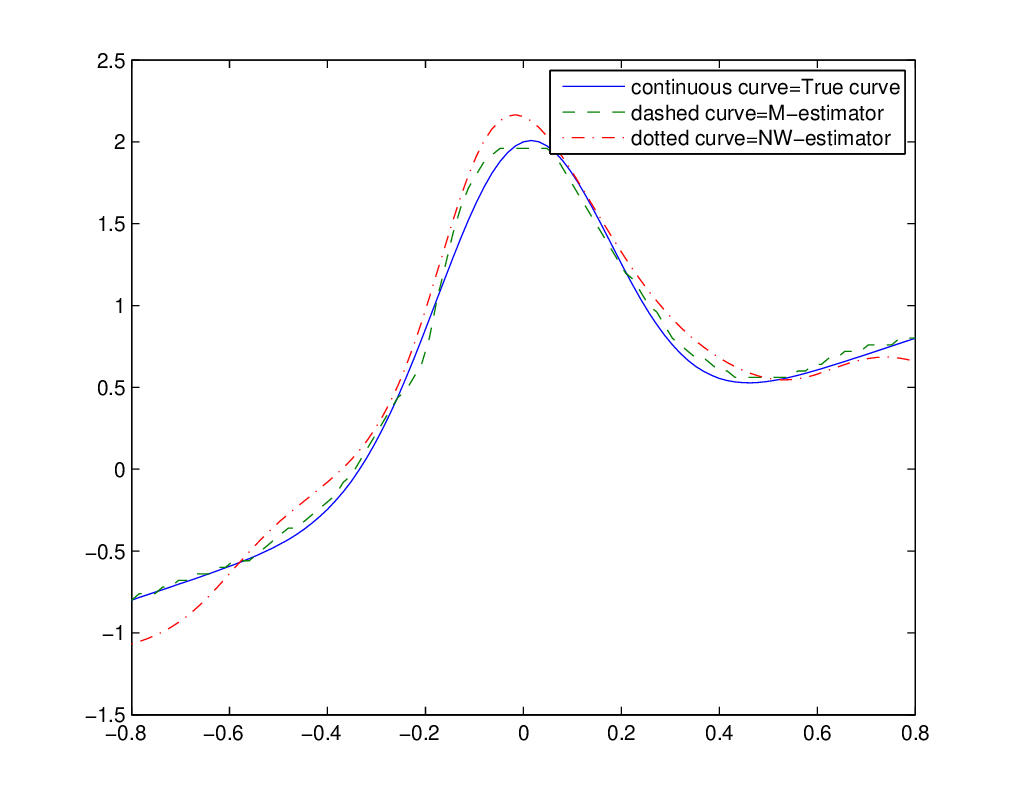} \\
\hspace{1cm}{\footnotesize{$n=300, CR=20\%$}} & \hspace{1cm}{\footnotesize{$n=300, CR=40\%$}}&\hspace{1cm}{\footnotesize{$n=300, CR=60\%$}}\\
\end{tabular}
\end{center}
\caption{Curves of the true regression function of model 1, $\widehat{\theta}_{\psi,n}(x)$ and $\widehat{\theta}_{n}^\star(x)$ obtained when $n=300$ and $CR=20\%, 40\%$ and $60\%$ respectively.} 
\label{n300}
\end{figure}

In Figures \ref{n300} and \ref{n800}, we plot the true regression function $m(\cdot)$ (from model 1) and its estimators $\widehat{\theta}_{\psi,n}(\cdot)$ and $\widehat{\theta}_n^\star(\cdot)$ where $\mbox{CR} = 20\%, 40\%$ and $60\%$ and $n=300$ and $800$ and $h_n=n^{-1/3}$. From Figures \ref{n300} and \ref{n800}, it can be seen that, despite the two estimators perform better for bigger sample size $n$ and for small censoring rates, the $M$-estimator stay more accurate than the NW one in all cases. Similar conclusions have been observed for model 2 and model 3.

Here, we want to point out that, although the Kaplan-Meier estimator has some drawback in the edges of the interval, it is nevertheless important of penalized by the survival lae of the censorship r.v. to limit the bad effect in edges. Nevertheless this seems very clear in the usual regression, on the other hand the case of the robust estimator, the effect is very limited by the robust function $\psi$, what appears clearly in Figures \ref{n300} and \ref{n800}.

\begin{figure}[h!]
\begin{center}
\begin{tabular}{lll}
\includegraphics[height=7cm,width=5.5cm]{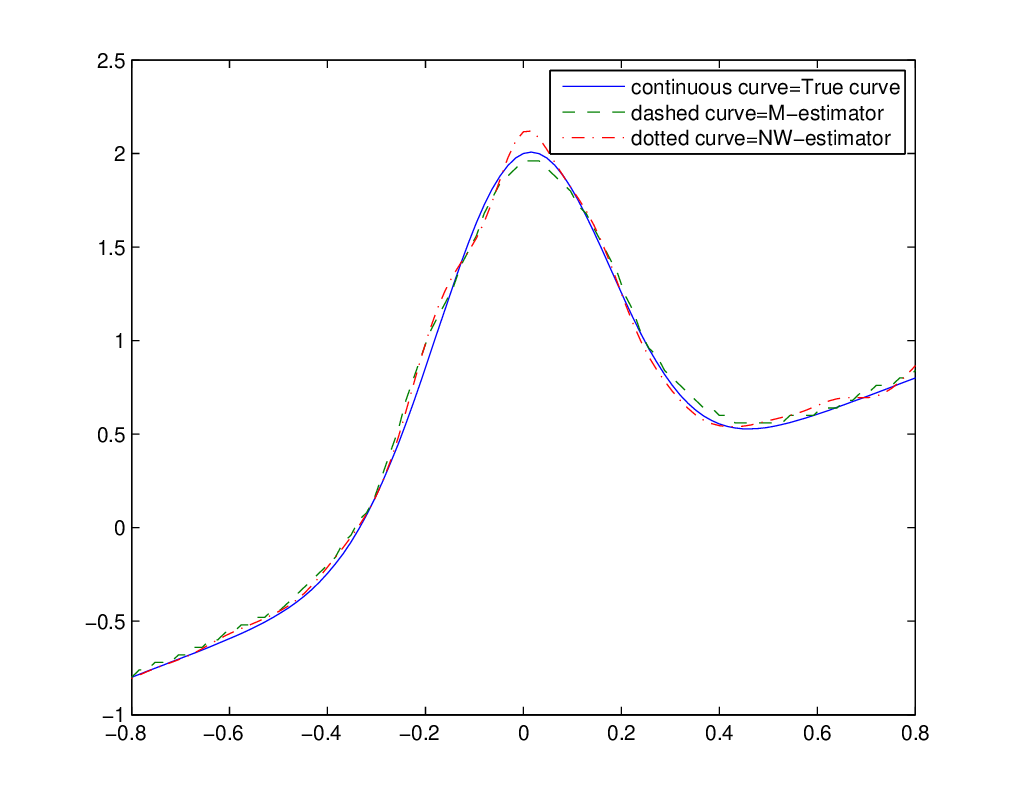} & \includegraphics[height=7cm,width=5.5cm]{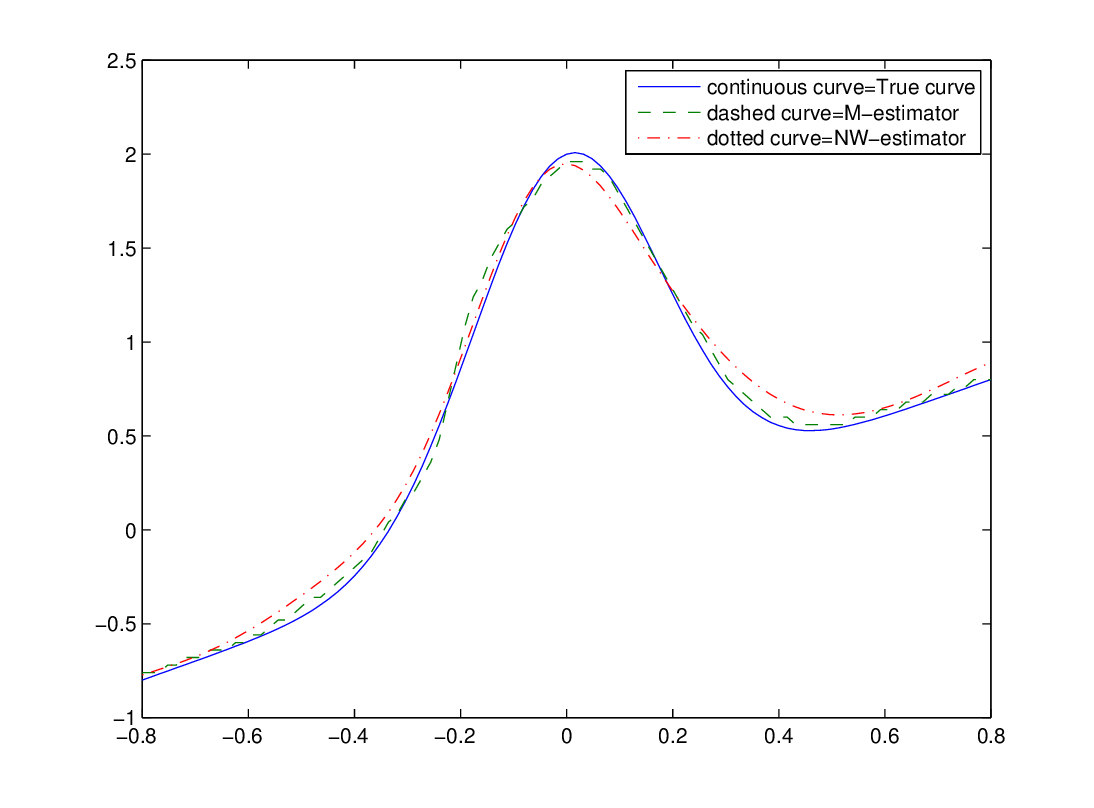}&
\includegraphics[height=7cm,width=5.5cm]{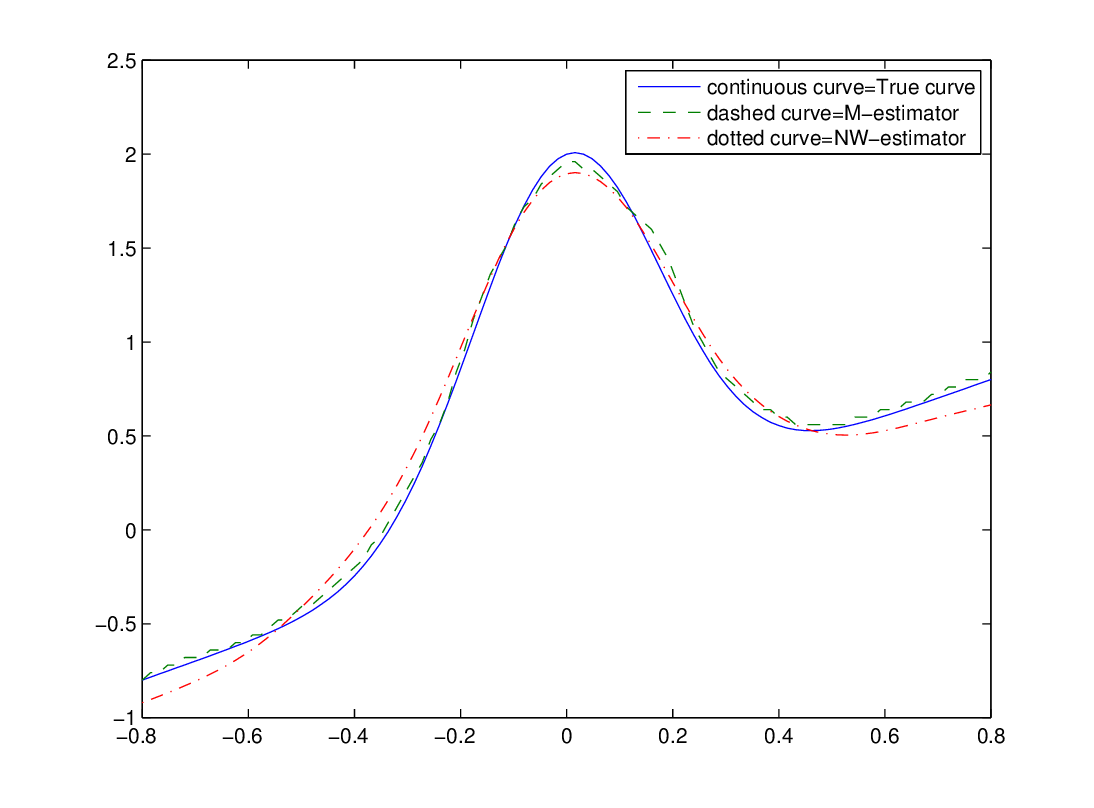} \\
\hspace{1cm}{\footnotesize{$n=800, CR=20\%$}} & \hspace{1cm}{\footnotesize{$n=800, CR=40\%$}}&\hspace{1cm}{\footnotesize{$n=800, CR=60\%$}}\\
\end{tabular}
\end{center}
\caption{Curves of the true regression function of model 1, $\widehat{\theta}_{\psi,n}(x)$ and $\widehat{\theta}_{n}^\star(x)$ obtained when $n=800$ and $CR=20\%, 40\%$ and $60\%$ respectively.}
\label{n800}
\end{figure}

\subsection{Simulation 2: on the robustness of the $M$-estimator}

\noindent  In this subsection, we are interested in the study of the robustness of the $M$-estimator
to the presence of outliers in the data. To reach this purpose, the model 1 is considered
and a sample of $n=800$ observations is generated. A percentage of $30\%$ of the data is
perturbed by multiplying the original data by a coefficient $k=5, 10$ and $20$ respectively.
 We generate $B=500$ replications  and the accuracy of the estimators is measured by the
 General Mean Square Error. Moreover, the robustness study of $\widehat{\theta}_{\psi,n}$
 and $\widehat{\theta}_{n}^\star$ has been considered when the censored rate $\mbox{CR}=20\%$
  and $60\%$ respectively. One can observe from Table \ref{table2} that $M$-estimator is less
  sensitive to the presence of outliers than the NW-estimator. The sensitivity to outliers
  of the last estimator increases significantly when the censoring rate increases.

\begin{table}
\caption{The global mean squared errors $\mbox{GMSE}(\widehat{\theta}_{\psi,n})$ and $\mbox{GMSE}(\widehat{\theta}_{n}^\star)$.}
\label{table2}
\centering
\begin{tabular}{lcccc}
\hline
$\mbox{CR}(\%)$ & $k$ & $\widehat{\theta}_{\psi,n}$ & $\widehat{\theta}_{n}^\star$ \\
\hline
& 5 & 1.78 & 3.92  \\
20 & 10& 1.84 & 10.65\\
& 20 & 1.91 & 34.01  \\
&&&&\\
& 5 & 1.91 & 4.43  \\
60 &10 &2.00 & 11.88 \\
& 20 & 2.01 & 38.22  \\
\hline
\end{tabular}
\end{table}

\newpage
\section{Proofs}\label{sec5}
\noindent  The proof of the results needs somme additional notations. We denote by
\begin{eqnarray}\label{ProofNotat1}
\overline{\widetilde{\Psi}}_N(x, \theta) &=& \frac{1}{n\mathbb{E}(\Delta_1(x))}
\sum_{i=1}^n \mathbb{E}\left\{\Delta_i(x) \frac{\delta_i \psi (Y_i- \theta)}{\overline{G}(Y_i)} \mid \mathcal{F}_{i-1}\right\} \quad \mbox{and}\nonumber\\
\overline{\widehat{\Psi}}_D(x) &=& (n\mathbb{E}(\Delta_1(x)))^{-1} \sum_{i=1}^n \mathbb{E}[\Delta_i(x)\mid \mathcal{F}_{i-1}].
\end{eqnarray}
Define  the {\it pseudo conditional bias} of $\widetilde{\Psi}_n(x,\theta)$  as
\begin{eqnarray}\label{bias}
B_n(x,\theta) := \frac{\overline{\widetilde{\Psi}}_N(x, \theta)}{\overline{\widehat{\Psi}}_D(x)} - \Psi(x,\theta),
\end{eqnarray}
and let

\begin{eqnarray}\label{ProofNotatR}
R_n(x,\theta) &:=& - B_n(x,\theta) [\widehat{\Psi}_D(x) - \overline{\widehat{\Psi}}_D(x)],
\end{eqnarray}
\begin{eqnarray}\label{ProofNotatR2}
 Q_n(x,\theta) &:=& [\widetilde{\Psi}_N(x,\theta) - \overline{\widetilde{\Psi}}_N(x, \theta)] - \Psi(x,\theta)[\widehat{\Psi}_D(x) - \overline{\widehat{\Psi}}_D(x)].
\end{eqnarray}

\noindent Therefore, we have from (\ref{ProofNotat1}), (\ref{bias}), (\ref{ProofNotatR})  and  (\ref{ProofNotatR2}) that :
\begin{eqnarray}\label{decomp1}
\widetilde{\Psi}_n(x,\theta) - \Psi(x,\theta) = B_n(x,\theta) + \frac{R_n(x,\theta) + Q_n(x,\theta)}{\widehat{\Psi}_D(x)}.
\end{eqnarray}

\noindent The proof of our results state in section 3, will be split into several lemmas given hereafter.

\vskip 2mm
\noindent  We begin  by  the  following  technical lemma, which  plays the same role as Bochner Theorem.
\begin{lemma}\label{lemlaib2005} Assume that Assumptions ({\bf A1}) and ({\bf A4}) hold true. For any $j\geq 1$, we have
\begin{itemize}
\item[$(i)$] $\mathbb{E}\left(\Delta_i^j(x) | \mathcal{F}_{i-1} \right) = c_0 h_{i,d}(x) b_n^d d \int_0^1 k^j(u) u^{d-1} du$ \ a.s. \\
\item[$(ii)$] $\mathbb{E}\left(\Delta_1^j(x) \right) = c_0 h(x) b_n^d d \int_0^1 k^j(u) u^{d-1} du$.
\end{itemize}
\end{lemma}
\begin{proof}.
Using Assumption ({\bf A4}) we can write
\begin{eqnarray*}
\mathbb{E}\left(\Delta_i^j(x) | \mathcal{F}_{i-1} \right) &=& \mathbb{E}\left(k^j(\|X_i - x \|_2/b_n) |  \mathcal{F}_{i-1}\right)\\
&=& \int_0^{b_n} k^j(u/b_n) d\mathbb{P}^{\mathcal{F}_{i-1}}\left(\|X_i - x \|_2 \leq u\right) \\
&=& \int_0^1 k^j(t) d\mathbb{P}^{F_{i-1}}\left(\| X_i -x\|_2/b_n \leq t \right)\\
&=& k^j(1) \mathbb{P}_{X_i}^{\mathcal{F}_{i-1}}(S_{b_n,x}) - \int_0^1 (k^j(u))^\prime \mathbb{P}_{X_i}^{\mathcal{F}_{i-1}}(S_{ub_n, x}) du.
\end{eqnarray*}
Finally, using Assumption ({\bf A1}), the proof of the result given by (i) can be concluded. The proof of (ii) follows from part (i) by considering $\mathcal{F}_{i-1}$ the trivial $\sigma$-field.
\end{proof}


\noindent The following exponential lemma is a tool that  be used when we deal with
difference martingale sequences (whose proof can be found in  \cite{LL11}).

 \begin{lemma}\label{lemA0} ( \cite{LL11}).
 Let $(Z_n)_{n\geq1}$ be a sequence of real martingale differences with respect
 to the sequence of $\sigma$-fields $(\mathcal{F}_n=\sigma(Z_1,\dots,Z_n))_{n\geq1}$,
 where $\sigma(Z_1, \dots, Z_n)$ is the $\sigma$-filed generated by the random variables
  $Z_1, \dots, Z_n$. Set $S_n=\sum_{i=1}^n Z_i.$ For any $p\geq2$ and any $n\geq1$,
  assume that there exist some nonnegative constants $\gamma_6$ and $d_n$ such that
\begin{eqnarray}
\mathbb{E}\left( Z_n^p\mid\mathcal{F}_{n-1}\right) \leq \gamma_6^{p-2}p!d_n^2\quad\quad\mbox{a.s.}
\end{eqnarray}
Then, for any $\epsilon>0$, we have
$$
\mathbb{P}\left(|S_n|>\epsilon \right)\leq 2\exp\left\{-\frac{\epsilon^2}{2(D_n+\gamma_6\epsilon)} \right\},
$$
where $D_n=\sum_{i=1}^n d_i^2.$
\end{lemma}

\noindent The following lemma gives the convergence  rate of the quantity $\widehat{\Psi}_D(x)$.
\begin{lemma}\label{lempsid}
Suppose that Assumptions ({\bf A1})-({\bf A4}) hold true, then we have
\begin{itemize}
\item[$(i)$] $ \lim_{n\rightarrow\infty}\widehat{\Psi}_D(x) = \lim_{n\rightarrow\infty} \overline{\widehat{\Psi}}_D (x) = 1, \ \
 a.s.
$
\item[$(ii)$] $\widehat{\Psi}_D(x) - \overline{\widehat{\Psi}}_D(x) = \mathcal{O}_{a.s.}
\left(\sqrt{\displaystyle\frac{\log n}{nb_n^d}} \right)$.
\end{itemize}
\end{lemma}
\begin{proof}.
First, observe that $\widehat{\Psi}_D(x)-1 = Z_{n,1}(x) + Z_{n,2}(x)$, where
$$Z_{n,1}(x) := \frac{1}{n\mathbb{E}(\Delta_1(x))} \sum_{i=1}^n \left(\Delta_i(x) - \mathbb{E}(\Delta_i(x) |
 \mathcal{F}_{i-1}) \right)$$ and
$$ Z_{n,2}(x) :=  \frac{1}{n\mathbb{E}(\Delta_1(x))} \sum_{i=1}^n \left[\mathbb{E}(\Delta_i(x)
| \mathcal{F}_{i-1}) - \mathbb{E}(\Delta_1(x)) \right].
$$
\noindent It can be easily seen that $\mathbb{E}(\Delta_i(x) | \mathcal{F}_{i-1}) - \mathbb{E}(\Delta_1(x))
 = c_0 b_n^d d \left[ h_i(x) - h(x)\right] \int_0^1 k(u) u^{d-1} du.$

\noindent Then Assumption ({\bf A2}) combined with Lemma \ref{lemlaib2005} allow to conclude that $Z_{n,2}(x) = o(1)$ a.s. as $n\rightarrow \infty.$

\noindent To deal with the term $Z_{n,1}(x)$  write it as follows
\begin{eqnarray*}
\displaystyle Z_{n,1}(x) = \frac{1}{n} \sum_{i=1}^n L_{n,i}(x),
\end{eqnarray*}
with  $ \displaystyle L_{n,i}(x) = \left[\Delta_i(x) - \mathbb{E}(\Delta_i(x) | \mathcal{F}_{i-1})
 \right]/\mathbb{E}(\Delta_1(x))$ is a martingale difference.

\noindent Now, the consistency  of the term  $Z_{n,1}(x)$ can follows from the application
of Lemma 1 in \cite{CLL13}. Therefore, let us check the condition on which this later can be applied.

\noindent Observe that
$$|L_{n,i}(x)| = \frac{|\Delta_i(x) - \mathbb{E}(\Delta_i(x) | \mathcal{F}_{i-1})|}{|\mathbb{E}
(\Delta_1(x))|}\leq \frac{2\overline{k}}{c_0 b_n^d d h(x) \int_0^1 k(u) u^{d-1} du} =: M,$$
\noindent where $\overline{k} := \sup_{x\in \mathbb{R}^+} k(x).$

\noindent On the other hand, making use of Lemma \ref{lemlaib2005}, one can see that
$$\displaystyle \mathbb{E}\left(L_{n,i}^2 (x) | \mathcal{F}_{i-1} \right) =
\frac{\int_0^1 k^2(u) u^{d-1} du}{c_0b_n^d d h^2(x) \left(\int_0^1 k(u) u^{d-1} du \right)^2}\ h_i(x) =: d_i^2.$$
\noindent Then using Lemma 1 in \cite{CLL13}, with $D_n = \sum_{i=1}^n d_i^2$, we obtain for any $\lambda > 0$

\begin{eqnarray*}
\displaystyle \mathbb{P}\left(\frac{1}{n} \Big|\sum_{i=1}^n L_{n,i} (x)\Big| \geq \lambda \right) &\leq & 2\exp\left\{-\frac{n^2 \lambda^2}{4 D_n + 2 M n \lambda} \right\}\\
& = & \displaystyle 2 \exp\left\{- \frac{n \lambda^2}{4\frac{D_n}{n} + 2 M \lambda} \right\}.
\end{eqnarray*}

\noindent Moreover, using Lemma \ref{lemlaib2005}, we get
$\displaystyle 4\frac{D_n}{n} + 2 M \lambda = \frac{C_\lambda(x)}{b_n^d},$ where
$$\displaystyle C_\lambda(x) = \frac{4 \left(\int_0^1 k^2(u) u^{d-1} du \right)
\left(\int_0^1 k(u) u^{d-1} du\right)^{-2} + 4 \lambda \overline{k}\left(\int_0^1 k(u) u^{d-1} du \right)}{c_0dh(x)}.$$
\noindent Consequently,
 $$
\displaystyle \mathbb{P}\left(\frac{1}{n} \Big|\sum_{i=1}^n L_{n,i} (x) \Big| \geq \lambda \right) \leq
 2 \exp\left\{-\frac{\lambda^2}{C_\lambda(x)}\ nb_n^d \right\}.
 $$
\noindent Taking into account the Assumption ({\bf A3}), the result is concluded by  Borel-Cantelli Lemma.
The quantity $\overline{\widehat{\Psi}}_D(x)$ could be treated in a similar way as
 $\widehat{\Psi}_D(x)$. Concerning the item (ii) one can use same arguments as the study of $Z_{n,1}(x)$.

\end{proof}

\noindent Lemma below establishes the convergence almost surely (with rate) of the conditional bias term $B_n$ and  the central term $R_n$ defined respectively in (\ref{bias}) and  (\ref{ProofNotatR}).
\begin{lemma}\label{lemBR} Assume that Assumptions ({\bf A1})-({\bf A4}), ({\bf A5})(iii) and ({\bf A6})(i) hold true, then we get
\begin{eqnarray}\label{B}
\sup_{\theta\in \mathcal{C}_{\psi,\tau}}|B_n(x,\theta)| = \mathcal{O}_{a.s.}(b_n^{d \alpha_1}),
\end{eqnarray}
and
\begin{eqnarray}\label{R}
\sup_{\theta\in \mathcal{C}_{\psi,\tau}}|R_n(x,\theta)| = \mathcal{O}_{a.s.}\left(b_n^{d \alpha_1} \sqrt{\frac{\log n}{n b_n^d}} \right).
\end{eqnarray}
\end{lemma}
\begin{proof}.  Recall that
$$
\displaystyle B_n(x,\theta) = \frac{\overline{\widetilde{\Psi}}_N(x,\theta) - \overline{\widehat{\Psi}}_D(x) \Psi(x,\theta)}{\overline{\widehat{\Psi}}_D(x)}.
$$
\noindent By a double conditioning  w.r.t. the $\sigma$-fields  $(\G_{i-1},\; T_i)$ and  $\F_{i-1}$, it follows from
  ({\bf A6})(i) that
\begin{eqnarray*}
\overline{\widetilde{\Psi}}_N(x,\theta) & = & \frac{1}{n\E(\Delta_1(x))} \sum_{i=1}^n \E\left\{\Delta_i(x) \E\left[  \frac{\delta_i\,\psi(Y_i- \theta)}{\overline{G}(Y_i)} | \G_{i-1}, T_i\right] |\F_{i-1}\right\}\\
&=& \frac{1}{n\E(\Delta_1(x))} \sum_{i=1}^n \E\left\{\Delta_i(x) \E\left[ \frac{\delta_i\,\psi(Y_i- \theta)}{\overline{G}(Y_i)} | X_i, T_i\right] |\F_{i-1}\right\}.
\end{eqnarray*}
\noindent Therefore, one gets
\begin{eqnarray*}
\overline{\widetilde{\Psi}}_N(x,\theta) - \overline{\widehat{\Psi}}_D(x) \Psi(x,\theta) = \frac{1}{n\E(\Delta_1(x))} \sum_{i=1}^n \E\left\{\Delta_i(x) \left[\Psi(X_i, \theta)- \Psi(x,\theta) \right] | \F_{i-1}\right\}.
\end{eqnarray*}

\noindent Making  use of the triangular inequality and Assumption ({\bf A5})(iii) we get
$$
\overline{\widetilde{\Psi}}_N(x,\theta) - \overline{\widehat{\Psi}}_D(x) \Psi(x,\theta) = \mathcal{O}_{a.s.}(b_n^{d \alpha_1}) \times \overline{\widehat{\Psi}}_D(x).
$$

\noindent Finally using Lemma \ref{lempsid} we deduce the consistency rate of $B_n(x,\theta)$ given by the equation (\ref{B}). Moreover, since $R_n(x,\theta) = - B_n(x,\theta) [\widehat{\Psi}_D(x) - \overline{\widehat{\Psi}}_D(x)]$, then equation (\ref{B}) and Lemma \ref{lempsid} allow to get the consistency rate of $R_n(x, \theta)$ given by the equation (\ref{R})
\end{proof}


\noindent The following lemma deals with the convergence rate of the numerator  $\widetilde{\Psi}_N(x,\theta)$ defined in (\ref{psodoestimate})  of the pseudo-estimator   $\widetilde{\Psi}_n(x,\theta)$.

\begin{lemma}\label{psitilde}
 Under Assumptions ({\bf A0})-({\bf A4}) and ({\bf A5})(iii), we have
\begin{eqnarray}
\sup_{\theta\in \mathcal{C}_{\psi,\tau}}|\widetilde{\Psi}_N(x,\theta) -
\overline{\widetilde{\Psi}}_N(x,\theta)| = \mathcal{O}_{a.s.}\left(  \sqrt{\frac{\log n}{n b_n^d}}\right).
\end{eqnarray}
\end{lemma}

\begin{proof}.  Since
$\displaystyle {\cal C}_{\psi,\tau}:=[\theta_{\psi}(x)-\tau, \; \theta_{\psi}(x)+\tau]$ is a compact set, it admits a covering by a finite number $d_n$ of balls $\mathcal{B}_j(t_j, \ell_n)$ centered at $t_j$, $1\leq j \leq d_n$ satisfies
 $\ell_n=n^{-1/2}$ and $d_n = \mathcal{O}(n^{-\alpha_2/2 })$.

\noindent Therefore,
\begin{eqnarray}\label{decomp}
\displaystyle \sup_{\theta\in \mathcal{C}_{\psi,\tau}}|\widetilde{\Psi}_N(x,\theta) - \overline{\widetilde{\Psi}}_N(x,\theta)| &\leq&
 \max_{1\leq j \leq d_n} \sup_{\theta\in \mathcal{B}_j} |\widetilde{\Psi}_N(x,\theta) - \widetilde{\Psi}_N(x,t_j)|
 + \max_{1\leq j \leq d_n} |\widetilde{\Psi}_N(x,t_j) -\overline{\widetilde{\Psi}}_N(x,t_j) |\nonumber\\
& & +  \max_{1\leq j \leq d_n} \sup_{\theta\in \mathcal{B}_j} |\overline{\widetilde{\Psi}}_N(x,t_j) -
\overline{\widetilde{\Psi}}_N(x,\theta)|\nonumber\\
&=:& \mathcal{J}_{n,1} +  \mathcal{J}_{n,2} +  \mathcal{J}_{n,3}.
\end{eqnarray}

\bigskip
\noindent{\bf Consistency of the first term $\mathcal{J}_{n,1}$}
\bigskip

\noindent Using Assumption ({\bf A6})(ii), we get
\begin{eqnarray*}
\displaystyle |\widetilde{\Psi}_N(x,\theta) - \widetilde{\Psi}_N(x,t_j)| &\leq& \frac{1}{n\E(\Delta_1(x))} \sum_{i=1}^n \Delta_i(x)\frac{\delta_i}{\overline{G}(Y_i)} \left|\psi(Y_i-\theta) - \psi(Y_i-t_j) \right|\\
&\leq &
\displaystyle \frac{\ell_n}{\overline{G}(\tau_F)} \widehat{\Psi}_D(x).
\end{eqnarray*}

\noindent Therefore,  since (by Lemma \ref{lempsid})  $\displaystyle
\widehat{\Psi}_D(x) = \mathcal{O}_{a.s.}(1)$  and $\overline{G}(\tau_F)> 0$, then we have
\begin{eqnarray}\label{J1}
\displaystyle \mathcal{J}_{n,1} = \mathcal{O}_{a.s.}\left(\sqrt{\frac{1}{n}}\right).
\end{eqnarray}

\bigskip
\noindent{\bf Consistency of the third term $\mathcal{J}_{n,3}$}
\bigskip

\noindent Using a double conditioning with respect to the $\sigma$-algebra
$\G_{i-1}$ and the definition of $\Psi(x,\theta)$ given by equation
(\ref{regression-censure}), one can easily obtain
\begin{eqnarray*}
\overline{\widetilde{\Psi}}_N(x,t_j) - \overline{\widetilde{\Psi}}_N(x,\theta) =
\frac{1}{n\E(\Delta_1(x))} \sum_{i=1}^n \E\left\{\Delta_i(x) \left[ \Psi(X_i, t_j)
- \Psi(X_i,\theta)\right] | \F_{i-1} \right\}.
\end{eqnarray*}
Then using Assumption ({\bf A5})(iii) and Lemma \ref{lempsid} we have
\begin{eqnarray}\label{J3}
\mathcal{J}_{n,3} = \mathcal{O}_{a.s.}(n^{-\alpha_2/2}).
\end{eqnarray}

\smallskip
\noindent{\bf Consistency of the second term $\mathcal{J}_{n,2}$}
\smallskip

\noindent First, observe that $\mathcal{J}_{n,2}$ can be written as
\begin{eqnarray*}
\displaystyle \mathcal{J}_{n,2} &=& \max_{1\leq j\leq d_n} \Big|\widetilde{\Psi}_N(x,t_j) -\overline{\widetilde{\Psi}}_N(x,t_j) \Big|\\
&=& \max_{1\leq j \leq d_n} \frac{1}{n\mathbb{E}(\Delta_1(x))} \Big|\sum_{i=1}^n S_{n,i}(x, t_j) \Big|,
\end{eqnarray*}
where $S_{n,i}(x, t_j) = \Delta_i(x)  \frac{\delta_i\,\psi(Y_i- t_j)}{\overline{G}(Y_i)} - \mathbb{E}\left[\Delta_i(x)  \frac{\delta_i\,\psi(Y_i- t_j)}{\overline{G}(Y_i)} | \mathcal{F}_{i-1}\right]$ is a martingale difference. Therefore, one can use Lemma \ref{lemA0} to obtain an exponential upper bound relative to the quantity $\widetilde{\Psi}_N(x,t_j) -\overline{\widetilde{\Psi}}_N(x,t_j)$.

\noindent Let us now check the conditions under which Lemma \ref{lemA0} is allowed to be applied.  We have, for any $p\in \mathbb{N}-\{0\}$, that
\begin{eqnarray*}
S_{n,i}^p(x, t_j) = \sum_{k=0}^p C_p^k \left(\frac{\delta_i}{\overline{G}(Y_i)} \psi(Y_i-t_j) \Delta_i(x) \right)^k (-1)^{p-k} \left[\mathbb{E}\left(\frac{\delta_i}{\overline{G}(Y_i)} \psi(Y_i-t_j) \Delta_i(x)  | \mathcal{F}_{i-1}\right) \right]^{p-k}.
\end{eqnarray*}

\noindent In view of Assumption ({\bf A6})(i), $\left[\mathbb{E}\left(\frac{\delta_i}{\overline{G}(Y_i)} \psi(Y_i-t_j) \Delta_i(x)  | \mathcal{F}_{i-1}\right) \right]^{p-k}$ is $\mathcal{F}_{i-1}$-measurable. Therefore,   using Jensen's  inequality,  one gets
\begin{eqnarray*}
\Big|\mathbb{E}\left( S_{n,i}^p(x,t_j) | \mathcal{F}_{i-1}\right) \Big|\leq \sum_{k=0}^p C_p^k \mathbb{E}\left\{\Big|\frac{\delta_i}{\overline{G}(Y_i)} \psi(Y_i-t_j) \Delta_i(x) \Big|^k | \mathcal{F}_{i-1} \right\} \times \mathbb{E}\left\{\Big|\frac{\delta_i}{\overline{G}(Y_i)} \psi(Y_i-t_j) \Delta_i(x) \Big|^{p-k} | \mathcal{F}_{i-1} \right\}
\end{eqnarray*}
Using now a double conditioning w.r.t. the $\sigma$-field $\mathcal{G}_{i-1}$, Assumptions ({\bf A0}) and ({\bf A6})(i), we get, for any $m\geq 1$, that
\begin{eqnarray*}
\mathbb{E}\left\{\Big|\frac{\delta_i}{\overline{G}(Y_i)} \psi(Y_i-t_j) \Delta_i(x) \Big|^{m} | \mathcal{F}_{i-1} \right\} &\leq& \frac{1}{(\overline{G}(\tau_F))^{m-1}} \mathbb{E}\left\{\Delta_i^m(x) \mathbb{E}\left[ \psi^m(Y_i-t_j) | X_i\right] | \mathcal{F}_{i-1} \right\}\\
&\leq & \frac{\gamma_3 m!}{(\overline{G}(\tau_F))^{m-1}} \mathbb{E}\left(\Delta_i^m(x) | \mathcal{F}_{i-1} \right).
\end{eqnarray*}
\noindent Making use of Lemma \ref{lemlaib2005}, we  obtain $\mathbb{E}\left(\Delta_i^m(x) | \mathcal{F}_{i-1} \right) = c_0 h_i(x) d b_n^d \int_0^1 k^m(u) u^{d-1} du$. It results then  from ({\bf A4})(ii) that
$$
\mathbb{E}\left(\Big|\frac{\delta_i \psi(Y_i-t_j)}{\overline{G}(Y_i)} \Delta_i(x) \Big|^k | \mathcal{F}_{i-1} \right)\; \mathbb{E}\left(\Big|\frac{\delta_i \psi(Y_i-t_j)}{\overline{G}(Y_i)} \Delta_i(x) \Big|^{p-k} | \mathcal{F}_{i-1} \right) \leq \gamma_3 h_i^2(x) b_n^{2d}.$$

\noindent Finally, we have   $\Big|\mathbb{E}\left( S_{n,i}^p(x,t_j) | \mathcal{F}_{i-1}\right) \Big| \leq 2^p h_i^2(x) b_n^{d}.$

\noindent By taking $d_i^2 = h_i^2(x) b_n^{d}$ and $D_n = \sum_{i=1}^n d_i^2 = b_n^d \sum_{i=1}^n h_i^2(x)$, one gets,
by Assumption ({\bf A2}), that   $n^{-1} D_n = b_n^d h(x)$ as $n\rightarrow \infty.$ Now one can use
Lemma \ref{lemA0} with $D_n = O_{a.s.}(nb_n^d)$ and $S_n = \sum_{i=1}^n S_{n,i}(x, t_j)$  to get, for any $\epsilon_0 >0,$
\begin{eqnarray*}
\mathbb{P}\left(|\mathcal{J}_{n,2}| > \epsilon_0 \sqrt{\frac{\log n}{n b_n^d}} \right) &
\leq & \sum_{j=1}^{d_n} \mathbb{P}\left(\Big|\widetilde{\Psi}_N(x,t_j) -
\overline{\widetilde{\Psi}}_N(x,t_j) \Big| > \epsilon_0 \sqrt{\frac{\log n}{n b_n^d}}  \right)\\
&=& \sum_{j=1}^{d_n} \mathbb{P}\left( \frac{1}{\mathbb{E}(\Delta_1(x))}\Big|\sum_{i=1}^n
S_{n,i}(x, t_j)\Big| > \epsilon_0 \sqrt{\frac{\log n}{n b_n^d}} \right)\\
&\leq & \sum_{j=1}^{d_n} 2 \exp\left\{ - \frac{(n\mathbb{E}(\Delta_1(x))\epsilon_0)^2 \frac{\log n}{n b_n^d}}{2D_n + 2\gamma_3n\mathbb{E}(\Delta_1(x))\epsilon_0\sqrt{\frac{\log n}{n b_n^d}}}\right\}\\
&\leq& 2d_n \exp\left\{-\gamma_3\epsilon_0^2 \log n \right\}\\
&=& n^{-\gamma_3\epsilon_0^2}\; n^{-\alpha_2/2}.
\end{eqnarray*}
\noindent Consequently, choosing $\epsilon_0$ such  that, the upper bound becomes a general term of a convergence Riemann series,  we get
 $\displaystyle \sum_{n\geq 1} \mathbb{P}\left(|\mathcal{J}_{n,2}| > \epsilon_0 \sqrt{\frac{\log n }{n b_n^d}} \right) < \infty.$
\noindent  The proof can be achieved by Borel-Cantelli Lemma.
\end{proof}
%

\noindent Lemma below study the uniform asymptotic rate of the quantity $Q_n(x,\theta)$.
\begin{lemma}\label{lemQ} Under assumptions ({\bf A1})-({\bf A4}) and ({\bf A5})(iii), we have
\begin{eqnarray}\label{Q}
\sup_{\theta\in \mathcal{C}_{\psi,\tau}}|Q_n(x,\theta)| = \mathcal{O}_{a.s.}\left(  \sqrt{\frac{\log n}{n b_n^d}}\right).
\end{eqnarray}
\end{lemma}

\begin{proof}. The proof of this Lemma can easily be  obtained from Lemmas  \ref{lempsid} and \ref{psitilde} .
\end{proof}

\begin{proof} {\bf of Proposition \ref{lem2prop1}}

Observe that
\begin{eqnarray*}
\Big| \widehat{\Psi}_n(x, \theta)  - \widetilde{\Psi}_n(x,\theta) \Big| &\leq & \frac{1}{n \mathbb{E}(\Delta_1(x)) \widehat{\Psi}_D(x)} \sum_{i=1}^n \Big|\Delta_i(x) \delta_i \psi(Y_i, \theta)\; \left(\frac{1}{\overline{G}_n(Y_i)} - \frac{1}{\overline{G}(Y_i)} \right)\Big|\\
&\leq & \frac{\sup_{t\in {\cal C}_{\psi,\tau}}|\overline{G}_n(t) - \overline{G}(t)|}{\overline{G}_n(\zeta)}\; \widetilde{\Psi}_n(x,\theta)
\end{eqnarray*}
\noindent Since $\overline{G}(\zeta) > 0$, in conjunction with the Strong Law of Large Numbers (SLLN)
and the Law of the Iterated Logarithm (LIL) on the censoring law (see \cite{DE00}),
 the result is an immediate consequence of decomposition (\ref{decomp1}),  Lemmas \ref{lempsid},   \ref{lemBR} and  \ref{lemQ}.
\end{proof}

\begin{proof} {\bf of Theorem \ref{prop1}}. Making use of the following decomposition

\begin{eqnarray}\label{decompProp}
\widehat{\Psi}_n(x, \theta) - \Psi(x, \theta) &=& \left( \widehat{\Psi}_n(x, \theta)  - \widetilde{\Psi}_n(x,\theta) \right)+
\left( \widetilde{\Psi}_n(x,\theta) - \Psi(x, \theta)\right)\nonumber\\
&=& \left( \widehat{\Psi}_n(x, \theta)  - \widetilde{\Psi}_n(x,\theta) \right)
+ B_n(x,\theta) + \frac{R_n(x,\theta) + Q_n(x,\theta)}{\widehat{\Psi}_D(x)}.
\end{eqnarray}

The proof  follows then as a consequence of Propostion \ref{lem2prop1}, Lemmas \ref{lempsid},   \ref{lemBR} and  \ref{lemQ}

\end{proof}
\begin{proof} {\bf of Theorem \ref{thm2}}.
Since,  $\Psi(x, \theta_{\psi}(x))=0$ and $\widehat{\Psi}_n(x, \widehat{\theta}_{\psi,n}(x))=0$, then
the Taylor's expansion of the function $\Psi(x,\cdot)$ around $\theta_{\psi}(x)$ leads to
\begin{eqnarray}\label{Taylor}
\widehat{\theta}_{\psi,n}(x) -  \theta_{\psi}(x) &=& \frac{\Psi(x, \widehat{\theta}_{\psi,n}(x))}{\Psi^\prime(x, \theta^{\star}_n)} \nonumber \\
&=& \frac{\Psi(x, \widehat{\theta}_{\psi,n}(x))
-\widehat{\Psi}_n(x, \widehat{\theta}_{\psi,n}(x))}{\Psi^\prime(x, \theta^{\star}_n)},
\end{eqnarray}
where $\theta^{\star}_n$ lies between $\widehat{\theta}_{\psi,n}(x)$ and $\theta_{\psi}(x).$  It follows from the Assumption (A6)(ii) that $ |\Psi^\prime(x, \theta^{\star}_n)|>\gamma$ and therefore
\begin{eqnarray}\label{eq1}
\Big|\widehat{\theta}_{\psi,n}(x) -  \theta_{\psi}(x) \Big|= \mathcal{O}\left(\sup_{\theta \in {\cal C}_{\psi,\tau}}\Big| \widehat{\Psi}_n(x, \theta)-\Psi(x, \theta) \Big|  \right).
\end{eqnarray}
The proof of Theorem \ref{thm2} can be then deduced from Theorem \ref{prop1}.
\end{proof}

\begin{proof} {\bf of Theorem \ref{norm1}}. The proof is based on the following decomposition
\begin{eqnarray*}
& & \left(nb_n^d\right)^{1/2}\left[\widehat{\Psi}_n(x, \theta) - \Psi(x, \theta)\right] \nonumber\\
&=& \left(nb_n^d\right)^{1/2}\left[\left(\widehat{\Psi}_n(x, \theta) - \widetilde{\Psi}_n(x,
 \theta)\right) + \left(\widetilde{\Psi}_n(x, \theta) - \overline{\widetilde{\Psi}}_n(x, \theta) \right)
  + \left(\overline{\widetilde{\Psi}}_n(x, \theta) - \Psi(x,\theta) \right)\right]\\
&=:& \mathcal{L}_{1,n} + \mathcal{L}_{2,n} + \mathcal{L}_{3,n}.
\end{eqnarray*}
By Proposition  \ref{lem2prop1} and assumption ({\bf A3})(ii), we have $\mathcal{L}_{1,n} = \mathcal{O}_{a.s.}\left(nb_n^d \sqrt{\log\log n/n}\right)=o_{a.s.}(1)$.
The term $\mathcal{L}_{n,3}$ is equal to $\sqrt{nb_n^d}B_n(x,\theta)$ which is $o_{a.s.}(1)$ in view of  assumption ({\bf A3})(ii) and Lemma  (\ref{lemBR}). Moreover,
observe that $\mathcal{L}_{n,2} = \sqrt{nb_n^d}\left( Q_n(x,\theta) - R_n(x,\theta)\right)/\widehat{\Psi}_D(x)$.
The quantity  $\sqrt{nb_n^d} R_n(x,\theta)$ converges almost surely to zero when $n$ goes to infinity, using the second part of  Lemma \ref{lemBR} combined
 with  assumption ({\bf A3})(ii).
Moreover, since  by Lemma \ref{lempsid},  $\lim_{n\rightarrow\infty}\widehat{\Psi}_D(x) = 1$ almost surely,
then using Slutsky's Theorem,  the asymptotic normality is given by the central term $\sqrt{nb_n^d}Q_n(x, \theta)$  which is the subject of the  following Lemma \ref{normQ}.
%
\end{proof}
\begin{lemma}\label{normQ}
Suppose that Assumptions ({\bf A1})-({\bf A8}) are satisfied, then we have
$$
\left(nb_n^d\right)^{1/2} Q_n(x,\theta) \stackrel{\mathcal{D}}{\longrightarrow} \mathcal{N}(0, \sigma^2(x, \theta)),
 \quad \quad \mbox{as}\;\; n\rightarrow\infty,
$$
where $ \sigma^2(x, \theta)$ is defined in Theorem \ref{norm1}.
\end{lemma}

\begin{proof}. Let us consider
$$
\eta_{ni} = \left(\frac{b_n^d}{n}\right)^{1/2} \left(\frac{\delta_i}{\overline{G}(Y_i)} \psi(Y_i-\theta) - \Psi(x,\theta)\right) \frac{\Delta_i(x)}{\mathbb{E}(\Delta_1(x))}
$$
and define $\xi_{ni} = \eta_{ni} - \mathbb{E}(\eta_{ni}|\mathcal{F}_{i-1})$. One can see that
\begin{eqnarray}\label{DD}
(nb_n^d)^{1/2} Q_n(x, \theta) = \sum_{i=1}^n \xi_{ni},
\end{eqnarray}
where, for any fixed $x\in \mathbb{R}^d$, the summands in (\ref{DD}) form a triangular
array of stationary martingale differences with respect to $\sigma$-field $\mathcal{F}_{i-1}$.
This allows us to apply the Central Limit Theorem for discrete-time arrays of real-valued martingales (as given in \cite{HH80}, page 23) to establish
the asymptotic normality of $Q_n(x, \theta)$. Therefore, we have to establish the following statements:
\begin{itemize}
\item[(a)] $\sum_{i=1}^n \mathbb{E}[\xi_{ni}^2|\mathcal{F}_{i-1}] \stackrel{\mathbb{P}}{\longrightarrow} \sigma^2(x,\theta),$
\item[(b)] $n \mathbb{E}[\xi_{ni}^2 \1_{[|\xi_{ni}| > \epsilon]}] = o(1)$ holds for any $\epsilon >0$ (Lindberg condition).
\end{itemize}

\noindent {\bf Proof of part (a)}. Observe that
$$
\left|\sum_{i=1}^n \mathbb{E}[\eta_{ni}^2|\mathcal{F}_{i-1}] - \sum_{i=1}^n \mathbb{E}[\xi_{ni}^2|\mathcal{F}_{i-1}] \right| \leq \sum_{i=1}^n \left(\mathbb{E}[\eta_{ni}|\mathcal{F}_{i-1}] \right)^2.
$$
\noindent By double conditioning with respect to ($\mathcal{G}_{i-1}, T_i)$  and ${\cal F}_{i-1}$ and using Assumptions ({\bf A2}) and ({\bf A5})(iii) and Lemma \ref{lemlaib2005}, we obtain
\begin{eqnarray}
\sum_{i=1}^n \left(\mathbb{E}[\eta_{ni}|\mathcal{F}_{i-1}]\right)^2 = \mathcal{O}_{a.s.}\left( b_n^{(2\alpha_1+1)d}\right).
\end{eqnarray}
\noindent Therefore, the statement of (a) follows then if we show that
\begin{eqnarray}\label{DDD}
\lim_{n\rightarrow\infty} \sum_{i=1}^n \mathbb{E}[\eta_{ni}^2 | \mathcal{F}_{i-1}] \stackrel{\mathbb{P}}{\longrightarrow} \sigma^2(x,\theta).
\end{eqnarray}
\noindent To prove (\ref{DDD}), observe that (by double conditioning) that
$$
\sum_{i=1}^n \mathbb{E}[\eta_{ni}^2|\mathcal{F}_{i-1}] = \frac{b_n^d/n}{(\mathbb{E}(\Delta_1(x)))^2} \sum_{i=1}^n \mathbb{E}\left\{\Delta_i^2(x) \mathbb{E}\left[\left( \frac{\delta_i}{\overline{G}(Y_i)}\psi(Y_i-\theta) - \Psi(x,\theta)\right)^2 | X_i\right]|\mathcal{F}_{i-1}\right\}.
$$
\noindent Using the definition of the conditional variance, one gets
\begin{eqnarray*}
\mathbb{E}\left[\left( \frac{\delta_i\psi(Y_i-\theta)}{\overline{G}(Y_i)} - \Psi(x,\theta)\right)^2 | X_i\right] &=& \mbox{Var}\left\{ \frac{\delta_i\psi(Y_i- \theta)}{\overline{G}(Y_i)}  | X_i \right\} + \left\{\mathbb{E}\left(\frac{\delta_i  \psi(Y_i- \theta) }{\overline{G}(Y_i)}| X_i\right) -\Psi(x,\theta) \right\}^2\\
&=& \mathcal{K}_{n,1} + \mathcal{K}_{n,2}.
\end{eqnarray*}
\noindent Using again a double conditioning with respect to $({\cal G}_{i-1}, T_i)$, Assumptions ({\bf A5})(iii) and ({\bf A2}) and Lemma \ref{lemlaib2005}, we obtain
$$
\frac{b_n^d/n}{(\mathbb{E}(\Delta_1(x)))^2} \sum_{i=1}^n \mathbb{E} [\Delta_i^2(x)\mathcal{K}_{n,2} | \mathcal{F}_{i-1}] = \mathcal{O}_{a.s.}\left( b_n^{2\alpha_1d}\right).
$$
\noindent Similarly one can show, by Assumptions ({\bf A2}), ({\bf A3}), ({\bf A5})(iii), ({\bf A7}) and Lemma \ref{lemlaib2005}, that
\begin{eqnarray*}
\lim_{n\rightarrow\infty} \frac{b_n^d/n}{(\mathbb{E}(\Delta_1(x)))^2} \sum_{i=1}^n \mathbb{E}[\Delta_i^2(x)\mathcal{K}_{n,1} | \mathcal{F}_{i-1}] &=& \left( \mathbb{M}(x,\theta) - \Psi(x,\theta)\right)\times\frac{\int_0^1 k^2(u) u^{d-1} du}{h(x) d \left(\int_0^1 k(u) u^{d-1} du \right)^2}\\
&=:& \sigma^2(x,\theta).
\end{eqnarray*}

\noindent {\bf Proof of part (b)}

\noindent The Lindeberg's condition results from Corollary 9.5.2 in \cite{CT98} which implies that
$$
n \mathbb{E}\left[\xi_{ni}^2 \1_{[|\xi_{ni}| > \epsilon]} \right] \leq 4n \mathbb{E}\left[\eta_{ni}^2 \1_{[|\eta_{ni}| >\epsilon/2]} \right].
$$
\noindent Let $a>1$ and $b>1$ such that  $1/a + 1/b =1$. Making use of H\"older and Markov inequalities one can write, for all $\epsilon >0$
$\mathbb{E}[\eta_{ni}^2 \1_{[|\eta_{ni}|> \epsilon/2]}] \leq \frac{\mathbb{E}|\eta_{ni}|^{2a}}{(\epsilon/2)^{2a/b}}$. Taking $C_0$ a positive constant and $2a=2+\varsigma$, one gets by Assumption ({\bf A8}) that
\begin{eqnarray*}
4n \mathbb{E}\left[ \eta_{ni}^2 \1_{[|\eta_{ni}|>\epsilon/2]}\right] &\leq& C_0\left( \frac{b_n^d}{n}\right)^{(2+\varsigma)/2}\frac{n\mathbb{E}[\Delta_1^{2+\varsigma}(x)]}{\left( \mathbb{E}(\Delta_1(x))\right)^{2+\varsigma}}\;\; [\widetilde{\mathbb{V}}_{2+\varsigma}(\theta|x) + o(1)].
\end{eqnarray*}
Finally Lemma \ref{lemlaib2005} allows one to write that $4n \mathbb{E}\left[ \eta_{ni}^2 \1_{[|\eta_{ni}|>\epsilon/2]}\right]  = \mathcal{O}_{a.s.}\left( \left( nb_n^d\right)^{-\varsigma/2}\right)$ which complete the proof by using Assumption ({\bf A3}).
\end{proof}

 \begin{proof} {\bf  of Theorem \ref{norm2}}.
 \noindent  Using a Taylor's expansion of order one of $\widehat{\Psi}_{n}(x,\cdot)$ around
 $\theta_{\psi}(x)$ and the definition of $\theta_\psi(x)$, we get
 $$
 \widehat{\Psi}_n(x, \theta_\psi(x)) - \Psi(x, \theta_\psi(x)) = \left(\theta_\psi(x)
 - \widehat{\theta}_{\psi,n}(x)\right) \widehat{\Psi}_n^\prime(x, \theta^\star_n)
 $$
 where $\theta^\star_n$ lies between $\theta_\psi(x)$ and $\widehat{\theta}_{\psi,n}(x).$ Then, we  have
 $$
\widehat{\theta}_{\psi,n}(x) -  \theta_\psi(x) = - \frac{ \widehat{\Psi}_n(x, \theta_\psi(x))
 - \Psi(x, \theta_\psi(x))}{\widehat{\Psi}_n^\prime(x, \theta^\star_n)}.
 $$
 Therefore,  the asymptotic normality given by Theorem \ref{norm2} can be stated using
 Theorem \ref{norm1} and the following lemma which provides the convergence
 in probability of the denominator term $\widehat{\Psi}_n^\prime(x, \theta^\star_n)$
 to $\Gamma_1(x,\theta_\psi(x))$.
 \end{proof}

 \begin{lemma}\label{lemnorm2} Under assumptions of Theorem \ref{norm2}, we have, uniformly in $\theta$,

$$ \widehat{\Psi}_n^\prime(x, \theta) \rightarrow \Gamma_1(x,\theta)\;\; \mbox{in probablility as}\;\; n\rightarrow\infty.$$
 \end{lemma}
 \begin{proof}. The proof of this lemma is based on the following decomposition
\begin{eqnarray}\label{DLT2}
|\widehat{\Psi}_n^\prime(x, \theta^\star) - \Gamma_1(x,\theta(x))| \leq |\widehat{\Psi}_n^\prime(x, \theta^\star) - \widehat{\Psi}_n^\prime(x, \theta(x))| + |\widehat{\Psi}_n^\prime(x, \theta(x)) - \Gamma_1(x,\theta(x))|
\end{eqnarray}
Concerning the first term, using the fact that $\delta_i$ is bounded by one and $\overline{G}_n(Y_i)$ is dominated by $\overline{G}_n(\zeta)$, then one can write
$$
|\widehat{\Psi}_n^\prime(x, \theta^\star) - \widehat{\Psi}_n^\prime(x, \theta_\psi(x))| \leq \sup_{t\in \mathcal{C}_{\psi,\tau}}\left|\frac{\partial \psi(t-\theta^\star)}{\partial \theta} - \frac{\partial \psi(t-\theta_\psi(x))}{\partial \theta}\right| \;\frac{1}{\overline{G}_n(\zeta)}.
$$
\noindent Because $\partial\psi(T-\theta)/\partial\theta$ is continuous at $\theta_\psi(x)$ uniformly in $t$, the use of Theorem \ref{thm2} and the convergence in probability of $\overline{G}_n(\zeta_F)$ to $\overline{G}(\tau_F)$ allows  to conclude that the first term of (\ref{DLT2}) converges in probability to zero.

Considering Assumptions ({\bf A1})-({\bf A4}), ({\bf A5})(iii) and ({\bf A6})(i), one can show, by using similar arguments as in the proof of Lemma \ref{lemBR}, that the second term in the right side of the inequality (\ref{DLT2}) converges almost surely to zero.
\end{proof}

\begin{proof} {\bf of Corollary \ref{corr}}. Observe that
\begin{eqnarray*}
\widehat{\Gamma}_{1,n}(x,\widehat{\theta}_{\psi,n}(x)) \left\{\frac{n\; \widehat{\mathbb{P}}_{X_i} (S_{b_n, x})\; d}{\widehat{\mathbb{M}}_n(x,
 \widehat{\theta}_{\psi,n}(x)) \int_0^1 k^2(u) u^{d-1} du} \right\}^{1/2}   \int_0^1 k(u) u^{d-1}du \times \left(\widehat{\theta}_{\psi,n}(x) - \theta_\psi(x) \right) &=&\\
\left\{\frac{\widehat{\mathbb{P}}_{X_i}(S_{b_n,x})}{b_n^d h(x)} \right\}^{1/2} \left\{ \frac{\mathbb{M}(x,\theta_\psi(x))}{\widehat{\mathbb{M}}_n(x, \widehat{\theta}_{\psi,n}(x))}\right\}^{1/2}\left\{ \frac{\widehat{\Gamma}_{1,n}(x, \widehat{\theta}_{\psi,n}(x))}{\Gamma_1(x,\theta_\psi(x))}\right\}&\times&\\
  \Gamma_1(x,\theta_\psi(x))\left\{ \frac{n b_n^d h(x) d}{\mathbb{M}(x,\theta_\psi(x)) \int_0^1 k^2(u)u^{d-1}du}\right\}^{1/2}\int_0^1k(u) u^{d-1}du \times \left( \widehat{\theta}_{\psi,n}(x)-\theta_\psi(x)\right). & &
\end{eqnarray*}
Making use of Theorem \ref{norm2}, we get
$$
 \Gamma_1(x,\theta_\psi(x))\left\{ \frac{n b_n^d h(x) d}{\mathbb{M}(x,\theta_\psi(x)) \int_0^1 k^2(u)u^{d-1}du}\right\}^{1/2}\int_0^1k(u) u^{d-1}du \times \left( \widehat{\theta}_{\psi,n}(x)-\theta_\psi(x)\right) \stackrel{\mathcal{D}}{\longrightarrow} \mathcal{N}(0, 1).
$$
Therefore, the Corollary \ref{corr} is established if we show that
$$
\left\{\frac{\widehat{\mathbb{P}}_{X_i}(S_{b_n,x})}{b_n^d h(x)} \right\}^{1/2} \left\{ \frac{\mathbb{M}(x,\theta_\psi(x))}{\widehat{\mathbb{M}}_n(x, \widehat{\theta}_{\psi,n}(x))}\right\}^{1/2}\left\{ \frac{\widehat{\Gamma}_{1,n}(x, \widehat{\theta}_{\psi,n}(x))}{\Gamma_1(x,\theta_\psi(x))}\right\} \stackrel{\mathbb{P}}{\longrightarrow} 1, \quad \mbox{as}\;\; n\rightarrow \infty.
$$
By the consistency of the empirical distribution function of $X$ and the decomposition given by (\ref{Nor2}), one obtains
$$
\frac{\widehat{\mathbb{P}}_{X_i}(S_{b_n,x})}{b_n^d h(x)} \stackrel{\mathbb{P}}{\longrightarrow} 1, \quad \mbox{as}\;\; n\rightarrow \infty.
$$
Since $\widehat{\theta}_{\psi,n}(x)$ is a consistent estimator of $\theta_\psi(x)$ (see Theorem \ref{thm2}), then it suffices to show that $\widehat{\mathbb{M}}_n(x, \theta)  \stackrel{\mathbb{P}}{\longrightarrow} \mathbb{M}(x, \theta) \quad \mbox{and} \quad\widehat{\Gamma}_{1,n}(x, \theta)  \stackrel{\mathbb{P}}{\longrightarrow} \Gamma_1(x, \theta) \;\;\mbox{as}\;\; n\rightarrow \infty$ which are  consequence of the previous results.
\end{proof}


\end{document}